\documentclass[11pt]{article}

\usepackage{mathpazo}
\usepackage{amsfonts}
\usepackage{amsmath}
\usepackage{graphicx}
\usepackage{latexsym}
\usepackage{mathabx}

\usepackage[margin=1.05in]{geometry}
  
\usepackage[T1]{fontenc}
\usepackage{times}
\usepackage{color,graphicx}
\usepackage{array}
\usepackage{enumerate}
\usepackage{amsmath}
\usepackage{amssymb}
\usepackage{amsthm}
\usepackage{pgfplots}
\usepackage{pgf}
\usepackage{tikz}
\usetikzlibrary{patterns}
\usepgfplotslibrary{patchplots} 
\usetikzlibrary{pgfplots.patchplots} 
\pgfplotsset{width=9cm,compat=1.5.1}

\usepackage{enumitem}
\usepackage[english]{babel}
\usepackage{amsthm}
\newtheorem{theorem}{Theorem}

\newtheorem{definition}[theorem]{Definition}
\newtheorem{lemma}[theorem]{Lemma}
\newtheorem{corollary}[theorem]{Corollary}
\newtheorem{remark}[theorem]{Remark}
\newtheorem{proposition}[theorem]{Proposition}

\newtheorem{example}[theorem]{Example}

\newcommand{\zero}{{\bf{0}}}
\newcommand{\ones}{{\bf{1}}}

\newcommand{\be}{\mathbf e}
\newcommand{\bv}{\mathbf v}
\newcommand{\bx}{\mathbf x}
\newcommand{\by}{\mathbf y}
\newcommand{\bz}{\mathbf z}
\newcommand{\bu}{\mathbf u}
\newcommand{\bw}{\mathbf w}

\usepackage{xcolor}
\usepackage{makeidx}
\usepackage[pdftex,linktocpage=true,colorlinks,citecolor=blue,linkcolor=blue,pagebackref]{hyperref}
\usepackage{hyperref}
\usepackage{cleveref}
\usepackage[alphabetic,backrefs]{amsrefs}

\def\C{\mathbb{C}}
\def\R{\mathbb{R}}

\begin{document}

\title{Perfect state transfer between real pure states}

\author{Chris Godsil\textsuperscript{1}, Stephen Kirkland\textsuperscript{2}, and Hermie Monterde\textsuperscript{2}
}

\maketitle

\begin{center}
{\textit{Dedicated to Agnes T.\ Paras on the occasion of her 60th birthday}}
\end{center}

\begin{abstract}
Pure states correspond to one-dimensional subspaces of $\mathbb{C}^n$ represented by unit vectors. In this paper, we develop the theory of perfect state transfer (PST) between real pure states with emphasis on the adjacency and Laplacian matrices as Hamiltonians of a graph representing a quantum spin network. We characterize PST between real pure states based on the spectral information of a graph and prove three fundamental results: (i) every periodic real pure state $\bx$ admits perfect state transfer with another real pure state $\by$, (ii) every connected graph admits perfect state transfer between real pure states, and (iii) for any pair of real pure states $\bx$ and $\by$ and for any time $\tau$, there exists a real symmetric matrix $M$ such that $\bx$ and $\by$ admit perfect state transfer relative to $M$ at time $\tau$. We also determine all real pure states that admit PST in complete graphs, complete bipartite graphs, paths, and cycles. This leads to a complete characterization of pair and plus state transfer in paths and complete bipartite graphs. We give constructions of graphs that admit PST between real pure states. Finally, using results on the spread of graphs, we prove that amongst all $n$-vertex simple unweighted graphs, the least minimum PST time between real pure states is attained by any join graph for the Laplacian case, while it is attained by the join of an empty graph and a complete graph of appropriate sizes for the adjacency case.
\end{abstract}

	


\section{Introduction}

Throughout, we assume that $G$ is a simple undirected connected graph on $n$ vertices with positive edge weights having adjacency matrix $A$ and Laplacian matrix $L$. We say that $G$ is \textit{unweighted} if every edge of $G$ has weight one. A \textit{continuous quantum walk} on $G$ is determined by the matrix 
\begin{equation}
\label{Eq:transmat}
U_M(t):=e^{itM}, \quad t\in\R,
\end{equation}
where $i^2=-1$ and $M$ is a real symmetric matrix called the \textit{Hamiltonian} associated with $G$. Here, $M$ is indexed by the vertices of $G$ such that $M_{u,v}=0$ if and only if there is no edge between $u$ and $v$. Note that $U_M(t)$ is a complex, symmetric and unitary matrix for each $t\in \R$. We write $U_M(t)$ as $U(t)$ if $M$ is clear from the context. Sometimes, we take $M$ to be $A$ or $L$. But unless otherwise stated, our results apply to any real symmetric matrix $M$ that respects the adjacencies of $G$. We denote the $m\times n$ all-ones matrix and $n\times n$ identity matrix by $J_{m,n}$ and $I_n$, respectively. We write these matrices as $J$ and $I$ if the context is clear.

A quantum state is represented by a positive semidefinite matrix with trace one, known as a \textit{density matrix}. If the initial state of our quantum walk is represented by a density matrix $D$, then the state $D(t)$ at time $t$ is given by $$U(t)DU(-t),$$ where $U(-t)=\overline{U(t)}^\top=U(t)^{-1}$ \cite{godsil2017real}. We say that \textit{perfect state transfer} (PST) occurs between two density matrices $D_1$ and $D_2$ if for some time $\tau>0$, we have $$D_2=U(\tau)D_1U(-\tau).$$ We say that \textit{perfect state transfer} (PST) occurs between two nonzero complex vectors $\bx$ and $\by$ if for some time $\tau>0$, there is a unit complex number $\gamma$ called \textit{phase factor} such that
\begin{equation}
\label{Eq:1}
U(\tau)\bx=\gamma\by.
\end{equation}
The minimum such $\tau$ is called the \textit{minimum PST time}.

A density matrix $D$ is a \textit{pure state} if the rank of $D$ is equal to one, and it is a \textit{real state} if all its entries are real. Thus, a real pure state $D$ can be written as $D=\frac{1}{\|\bx\|^2}\bx\bx^\top$ for some nonzero vector $\bx\in\R^n$, where $\|\cdot\|$ is the Euclidean norm. For each $\bx\in\R^n$ with $\bx\neq \zero$, we let $D_{\bx}:=\frac{1}{\|\bx\|^2}\bx\bx^\top$ denote the real pure state associated with $\bx$. Since the equation $D_{\by}=U(\tau)D_{\bx}U(-\tau)$ is equivalent to $U(\tau)\bx=\gamma\by$ for some unit $\gamma\in\C$, the existence of perfect state transfer between real density matrices $D_{\bx}$ and $D_{\by}$ is equivalent to perfect state transfer between the real vectors $\bx$ and $\by$. Thus, we abuse terminology and also refer to $\bx$ as a real pure state. If $u$ and $v$ are two vertices in $G$ and $s$ is a non-zero real number, then a real pure state of the form $\bx=\be_u$ is called a \textit{vertex state}, while $\bx=\be_u+s\be_v$ is called an \textit{$s$-pair state} \cite{kim2024generalization}. In particular, a $(-1)$-pair state is called a \textit{pair state} and a $1$-pair state is called a \textit{plus state}. Perfect state transfer between vertex states has been extensively studied, see \cite{christandl2004,Coutinho2014,godsil2012state,kay2010perfect,kendon2011perfect} for surveys and \cite{soffia2022state,bhattacharjya2024quantum,coutinho2024no} for more recent work. On the other hand, perfect state transfer between pair states, between plus states, and between $s$-pair states has only been investigated recently \cite{bernard2025quantum,Chen2020PairST,kim2024generalization,pal2024quantum}. As $U(\tau)$ is unitary, $U(\tau)\bx=\gamma\by$ implies that $\|\bx\|=\|\by\|$. Thus, to simplify our discussion, we examine perfect state transfer between real vectors with the same length, in lieu of density matrices. Throughout, we assume that $\bx$ and $\by$ are real vectors with $\|\bx\|=\|\by\|\neq 0$. 

In this paper, we develop the theory of perfect state transfer between real pure states in weighted graphs. In Sections \ref{Sec:esupp} and \ref{Sec:sc}, we extend the concept of eigenvalue supports and strong cospectrality to real pure states. Section \ref{Sec:per} deals with periodicity. In particular, we show that periodicity of a real pure state with nonnegative entries is relatively rare whenever $M$ has nonnegative rational entries, a result that can be viewed as an extension of the relative rarity of vertex periodicity \cite{godsil2010can}. We also provide a simple formula for calculating the minimum period of a periodic real pure state. We devote Section \ref{Sec:pst} to a characterization of perfect state transfer between real pure states, which extends a characterization of vertex perfect state transfer due to Coutinho \cite{Coutinho2014}. We establish three important facts in this section: (i) every periodic real pure state $\bx$ admits perfect state transfer with another real pure state $\by$, (ii) every connected graph admits perfect state transfer between real pure states, and (iii) for any pair of real pure states $\bx$ and $\by$ and for any time $\tau$, there exists a real symmetric matrix $M$ such that $\bx$ and $\by$ admit perfect state transfer relative to $M$ at time $\tau$. In Sections \ref{Sec:cycle}, \ref{Sec:paths} and \ref{Sec:cbp}, we characterize perfect state transfer between real pure states in complete graphs, cycles, paths and complete bipartite graph. As a consequence, we obtain a characterization of pair and plus state transfer in paths and complete bipartite graphs. While only a few cycles and paths admit vertex perfect state transfer, it turns out that there are infinite families of such graphs that admit perfect state transfer between real pure states. Sections \ref{Sec:cartprod} and \ref{Sec:joins} are dedicated to constructions of graphs with PST between real pure states. In Section \ref{Sec:minpst}, we utilize results on the spread of graphs to establish that amongst all $n$-vertex unweighted graphs, the minimum PST time between real pure states for the Laplacian case is attained by any join graph, while for the adjacency case, it is attained by the join of an empty graph and a complete graph of appropriate sizes, provided that $n$ is sufficiently large. Finally, in Section \ref{Sec:sens}, we determine closed form expressions for $\frac{d^kf}{dt^k}\rvert_{\tau}$, where $f(t)=|\by^\top U(t)\bx|^2$ is the fidelity of transfer between real vectors $\bx$ and $\by$, and $f(\tau)=1$. We give sharp bounds on $\frac{d^2f}{dt^2}\rvert_{\tau}$, which describes the sensitivity of the fidelity at time $\tau$ with respect to the readout time.

Let $\bx\neq \zero$ be a real vector. Observe that for any time $t>0$, there is perfect state transfer between $\bx$ and $\gamma U(t)\bx$ for some unit $\gamma\in \C$.
But as $U(t)$ has complex entries, we are not guaranteed that $\gamma U(t)\bx$ has real entries. Now, suppose $\gamma U(\tau)\bx$ has real entries. In this case, a result of Godsil implies that perfect state transfer between $\bx$ and $\gamma U(\tau)\bx$ is monogamous \cite[Corollary 5.3]{godsil2017real}. That is, if there is perfect state transfer between $\bx$ and another vector $\by \in \R^n$, then $\by=\gamma U(\tau)\bx$. If we also assume that $\tau$ is the minimum PST time between $\bx$ and $\by=\gamma U(\tau)\bx$, then every PST time is an odd multiple of $\tau$ \cite[Lemma 5.2]{godsil2017real}. Hence, if $t=(2k+1)\tau$, then $$\gamma U(t)\bx=\gamma U(\tau)^{2k}U(\tau)\bx= U(\tau)^{2k}\by=\gamma^{-2k}\by.$$ That is, $\by=\gamma^{2k+1} U(t)\bx$. From this, we deduce that for every $\gamma\in\C$, there is at most one real vector in the set $\{\gamma U(t)\bx:t>0\}$ distinct from $\pm\bx$. This demonstrates that the existence of perfect state transfer between real pure states is a special occurrence, and therefore warrants an investigation.

Our work provides a spectral framework for studying perfect state transfer between real pure states that unifies the study of perfect state transfer between vertices and $s$-pair states. We recover known results about perfect state transfer between vertex states, plus states and pair states, and produce new instances of pair and plus state transfer.

\section{Eigenvalue supports}
\label{Sec:esupp}

Let $M$ be a Hamiltonian of a graph $G$. Since $M$ is real symmetric, we may write $M$ in its spectral decomposition:
\begin{equation}
\label{Eq:specdec}
    M=\sum_{j=1}^k\lambda_jE_j,
\end{equation} 
where $\lambda_1, \ldots, \lambda_k$ are the distinct eigenvalues of $M$ with corresponding eigenprojection matrices $E_1, \ldots, E_k$. Combining equations (\ref{Eq:transmat}) and (\ref{Eq:specdec}), we obtain a spectral decomposition of $U(t)$ given by
\begin{equation}
\label{Ut}
U(t)=\sum_{j=1}^ke^{it\lambda_j}E_j.
\end{equation}
Thus, the eigenvalues and eigenvectors of $M$ completely determine the behaviour of the quantum walk.

Let $\bx\in\R^n$ with $\bx\neq \zero$. The \textit{eigenvalue support} of $\bx$ relative to $M$, denoted $\sigma_{\bx}(M)$, is the set
\begin{equation*}
\sigma_{\bx}(M)=\{\lambda_j:E_j\bx\neq \zero\}.
\end{equation*}
Note that $\sigma_{\bx}(M)$ is always nonempty. If $|\sigma_{\bx}(M)|=1$, then $\bx$ is called a \textit{fixed state} relative to $M$.

\begin{proposition}
\label{Prop:S}
Let $\bx\in\R^n\backslash\{\zero\}$. If $S$ is a subset of the set of distinct eigenvalues of $M$, then $\sigma_{\bx}(M)=S$ if and only if $\bx=\sum_{j\in S}\bu_j$, where each $\bu_j$ is a real eigenvector associated with an eigenvalue $\lambda_j\in S$.
\end{proposition}

\begin{proof}
Let $\sigma_{\bx}(M)=S$. If $\lambda_j\in S$, then $\bu_j:=E_j\bx$ is an eigenvector for $M$ associated with $\lambda_j$. As the $E_j$s sum to identity, we obtain $\bx=I\bx=\sum_{\lambda_j\in S}E_j\bx=\sum_{\lambda_j\in S}\bu_j$. The converse is straightforward.
\end{proof}

We let $\phi(M,t)$ denote the characteristic polynomial of $M$ in the variable $t$. Adapting the same proof of Proposition 2.4 in \cite{kim2024generalization} yields a more general result.

\begin{lemma}
\label{Lem:suppalgconj}
Let $\bx\in\R^n\backslash\{\zero\}$ and suppose $a\bx$ has rational entries for some nonzero $a\in\R$. If $\phi(M,t)$ has integer coefficients, then $\sigma_{\bx}(M)$ is closed under taking algebraic conjugates.
\end{lemma}

If $S$ is a singleton set in Proposition \ref{Prop:S}, then we get the following result.

\begin{proposition}
\label{Prop:fixedstatechar}
A vector $\bx\in\R^n\backslash\{\zero\}$ is a fixed state if and only if $\bx$ is an eigenvector for $M$ associated with the lone eigenvalue in $\sigma_{\bx}(M)$.
\end{proposition}

If $\bx$ is a fixed state, then (\ref{Ut}) implies that $U(t)\bx=e^{it\lambda}\bx$ for any time $t$, where $\lambda$ is an eigenvalue of $M$ with associated eigenvector $\bx$. For connected graphs in particular, a vertex state $\be_u$ is not an eigenvector for $M$, and hence is not a fixed state.

\begin{example}
Some examples of fixed pure states include
(i) $\bx\in\operatorname{span}\{\ones\}$, if $G$ is regular or $M=L$, (ii)
$\bx\in\operatorname{span}\{\bv\}$, if $\bv$ is a Perron eigenvector for $M=A$, and (iii) $\bx\in\operatorname{span}\{\be_u-\be_v\}$, if  $u,v$ are twins in $G$.
\end{example}

The \textit{covering radius of a set} $S\subseteq V(G)$ is the least nonnegative integer $r$ such that each vertex of $G$ is at distance at most $r$ from $S$. The \textit{covering radius of a vector} $\bx\in\R^n$ is defined to be the covering radius of the set $S=\{u\in V(G):\bx^\top\be_u\neq 0\}$. We state Lemma 4.1 in \cite{godsil2012controllable}.

\begin{lemma}
\label{Lem:supp}
Suppose $\bx\in\R^n\backslash\{\zero\}$ is not a fixed state and has covering radius $r$. 
If $M$ and $\bx$ are entrywise nonnegative, then $|\sigma_{\bx}(M)|\geq r+1$. 
\end{lemma}

\begin{remark}
If $G$ is a primitive strongly regular graph and $S=\{u,v\}$, where $u$ and $v$ are adjacent, then $r=2$ but $\bx=\be_u-\be_v$ satisfies $|\sigma_{\bx}(M)|=2$. Thus, Lemma \ref{Lem:supp} need not hold if $\bx$ is not entrywise nonnegative.
\end{remark}

\section{Periodicity}
\label{Sec:per}

\begin{definition}
Suppose $\bx\in \R^n\backslash\{\zero\}$ is not a fixed state. We say that $\bx$ is \textit{periodic} in $G$ (relative to $M$) if there is a time $\tau>0$ such that
\begin{equation*}
    U(\tau)\bx=\gamma\bx
\end{equation*}
for some unit $\gamma\in\mathbb{C}$. The minimum such $\tau$ is called the \textit{minimum period} of $\bx$, which we denote by $\rho$. 
\end{definition}

A set $S\subseteq\R$ with at least two elements satisfy the \textit{ratio condition} if $$\frac{\lambda_p-\lambda_q}{\lambda_r-\lambda_s}\in\mathbb{Q}$$ for all $\lambda_p,\lambda_q,\lambda_r,\lambda_s\in S$ with $\lambda_r \ne \lambda_s$. Note that $S$ automatically satisfies the ratio condition whenever $|S|=2$. 

\begin{theorem}
\label{Thm:rc}
A vector $\bx\in \R^n$ is periodic in $G$ if and only if $\sigma_{\bx}(M)$ satisfies the ratio condition. If we also assume that $|\sigma_{\bx}(M)|\geq 3$ and $\sigma_{\bx}(M)$ is closed under algebraic conjugates,
then $\bx$ is periodic if and only if either (i) $\sigma_{\bx}(M)\subseteq\mathbb{Z}$ or (ii) each $\lambda_j\in \sigma_{\bx}(M)$ is of the form $\lambda_j=\frac{1}{2}(a+b_j\sqrt{\Delta})$, where $a,b_j,\Delta$ are integers and $\Delta>1$ is square-free.
\end{theorem}

\begin{proof}
This follows from \cite[Corollary 7.3.1]{Coutinho2021}, and \cite[Theorem 7.6.1]{Coutinho2021}. 
\end{proof}

The following is straightforward from Theorem \ref{Thm:rc}.

\begin{corollary}
\label{Cor:specgap}
Suppose $\bx\in \R^n$ such that
$|\sigma_{\bx}(M)|\geq 3$ and $\sigma_{\bx}(M)$ is closed under taking algebraic conjugates. If $\bx$ is periodic,
then the elements in $\sigma_{\bx}(M)$ differ by at least one. 
\end{corollary}

We now turn our attention to the minimum period.

\begin{lemma}\label{Lem:minperiod}
Let $G$ be a graph and $\bx\in\R^n\backslash\{\zero\}$ with $\sigma_{\bx}(M)=\{\lambda_1,\ldots,\lambda_m\}$, where $\lambda_1>\lambda_2$.
\begin{enumerate}
\item If $m=2$, then $\bx$ is periodic in $G$ with $\rho=\frac{2\pi}{\lambda_1-\lambda_2}$
\item If $m\geq 3$ and $\bx$ is periodic in $G$, then $\rho=\frac{2\pi q}{\lambda_1-\lambda_2}$, where $q=\operatorname{lcm}(q_3,\ldots,q_{m})$ and the $p_j$'s and $q_j$'s are integers such that $\frac{\lambda_1-\lambda_j}{\lambda_1-\lambda_2}=\frac{p_j}{q_j}$ and $\operatorname{gcd}(p_j,q_j)=1$.
\end{enumerate}
\end{lemma}

\begin{proof}
This follows from a simple extension of the proof of \cite[Theorem 5]{Monterde2022}.
\end{proof}

\begin{corollary}
\label{Lem:minperiod1}
In Lemma \ref{Lem:minperiod}(2), if we assume in addition that $\sigma_{\bx}(M)$ is closed under taking algebraic conjugates, then $\rho=\frac{2\pi}{g\sqrt{\Delta}}$, where $g=\operatorname{gcd}(\frac{\lambda_1-\lambda_2}{\sqrt{\Delta}},\frac{\lambda_1-\lambda_3}{\sqrt{\Delta}},\ldots,\frac{\lambda_1-\lambda_m}{\sqrt{\Delta}})$.
\end{corollary}

\begin{proof}
By Theorem \ref{Thm:rc}, we may write each $\lambda_j=\frac{1}{2}(a+b_j\sqrt{\Delta})$, where $\Delta\geq 1$. 
Thus, each $q_j$ in Lemma \ref{Lem:minperiod}(2) can be written as $q_j=\frac{b_1-b_2}{2g_j}$, where $g_j=\operatorname{gcd}\left(\frac{b_1-b_j}{2},\frac{b_1-b_2}{2}\right)$. Therefore, $q=\operatorname{lcm}(q_3,\ldots,q_m)=\operatorname{lcm}\left(\frac{b_1-b_2}{2g_3},\ldots,\frac{b_1-b_2}{2g_m}\right)=\frac{b_1-b_2}{2g}=\frac{\lambda_1-\lambda_2}{g\sqrt{\Delta}}$. Applying Lemma \ref{Lem:minperiod}(2) then yields $\rho=\frac{2\pi q}{\lambda_1-\lambda_2}=\frac{2\pi}{g\sqrt{\Delta}}$.
\end{proof}

\begin{theorem}
\label{Thm:covrad}
Let $M$ and $\bx\in\R^n\backslash\{\zero\}$ be entrywise nonnegative. 

\begin{enumerate}
\item If $|\sigma_{\bx}(M)|=2$, then $\bx$ is periodic and the covering radius of $\bx$ is at most one.
\item Suppose $\bx$ is periodic relative to $M$. If $|\sigma_{\bx}(M)|\geq 3$ and $\sigma_{\bx}(M)$ is closed under taking algebraic conjugates, then the covering radius of $\bx$ is at most $2k$, where $k$ is the maximum row sum of $M$.
\end{enumerate}
\end{theorem}

\begin{proof}
Let $r$ be the covering radius of $\bx$. Since $M$ and $\bx$ are entrywise nonnegative, Lemma \ref{Lem:supp} yields $r+1\leq |\sigma_{\bx}(M)|$. Combining this with Lemma \ref{Lem:minperiod}(1) gives us (1). Now, let $\rho(M)$ denote the spectral radius of $M$. As $|\sigma_{\bx}(M)|\geq 3$, $\sigma_{\bx}(M)$ is closed under taking algebraic conjugates and $\bx$ is periodic, Corollary \ref{Cor:specgap} yields $|\sigma_{\bx}(M)|\leq 2\rho(M)+1$. As $\rho(M)\leq k$, we get $r+1\leq |\sigma_{\bx}(M)|\leq 2\rho(M)+1\leq 2k+1,$ which yields the desired result in (2).
\end{proof}

\begin{remark}
\label{Rem:lapcovrad}
Theorem \ref{Thm:covrad} applies to $A$. Since we may take $M=kI-L$, Theorem \ref{Thm:covrad} also applies to $L$. 
\end{remark}

\begin{theorem}
\label{Thm:rare}
For each $k>0$, there are only finitely many connected graphs with positive integer weights and maximum degree at most $k$ such that a vector $\bx\neq 0$ with nonnegative rational entries is periodic relative to $A$ or $L$.
\end{theorem}

\begin{proof}
Assume that $G$ is a connected unweighted graph with maximum degree $k$ and let $\bx$ be a vector with nonnegative rational entries that is periodic relative to $M$. By Lemma \ref{Lem:suppalgconj}, $\sigma_{\bx}(M)$ is closed under taking algebraic conjugates. Let $r$ be the covering radius of $\bx$. Since $\bx$ is not fixed, we have $|\sigma_{\bx}(M)|\geq 2$, and so applying Theorem \ref{Thm:covrad} to $M\in\{A, kI-L\}$ yields $r\leq 2k$. As $k$ is fixed, there are only finitely many connected unweighted graphs with degree at most $k$ and the covering radius of $\bx$ is bounded above by $2k$. This remains true if we assign positive integer weights to $G$.
\end{proof}

Theorem \ref{Thm:rare} generalizes Godsil's result on periodic vertex states \cite[Corollary 6.2]{godsil2010can} and Kim et.\ al's results on periodic $s$-pair states with nonnegative rational entries \cite[Corollary 3.5]{kim2024generalization}. We also note that Theorem \ref{Thm:rare} need not apply if $\bx$ has a positive and a negative entry. See \cite{pal2024quantum} for an infinite family of trees with maximum degree three admitting PST between pair states. 

\begin{remark}
The argument in the proof of Theorem \ref{Thm:rare} applies when the Hamiltonian taken is the signless Laplacian matrix. In this case, we obtain the bound $r\leq 4k$ in the above proof in lieu of $r\leq 2k$. By Lemma \ref{Lem:suppalgconj}, Theorem \ref{Thm:rare} applies to entrywise nonnegative vectors $\bx\neq \zero$ whenever $a\bx$ has rational entries for some $a>0$.
\end{remark}

\section{Strong cospectrality}
\label{Sec:sc}

\begin{definition}
\label{scdef}
Let $\bx,\by\in \R^n\backslash\{\zero\}$ with $\by\neq\pm \bx$ and $\|\bx\|=\|\by\|$. We say that $\bx$ and $\by$ are \textit{strongly cospectral} (relative to $M$) if for each $\lambda_j \in \sigma_{\bx}(M)$, either  
$E_j\bx= E_j\by$ or $E_j\bx= -E_j\by$.
\end{definition}

The above definition allows us to partition $\sigma_{\bx}(M)$ into two sets $\sigma_{\bx,\by}^{+}(M)$ and $\sigma_{\bx,\by}^{-}(M)$ given by
\begin{equation*}
\sigma_{\bx,\by}^{+}(M)=\{\lambda_j:E_j\bx=E_j\by\neq \zero\}\quad \text{and} \quad \sigma_{\bx,\by}^{-}(M)=\{\lambda_j:E_j\bx=-E_j\by\neq \zero\}.
\end{equation*}
Thus, if $\bx$ is involved in strong cospectrality, then $|\sigma_{\bx}(M)|\geq 2$. Consequently, a fixed state cannot be strongly cospectral with another pure state by Proposition \ref{Prop:fixedstatechar}.

The proof of our next result is analogous to that of Theorem 3.1 and Lemma 10.1 in \cite{GodsilSmith2024}, and so we omit it here.

\begin{lemma}
\label{prop:}
The following are equivalent.
\begin{enumerate}
\item For all $j$, $\bx^\top E_j\bx=\by^\top E_j\by$.
\item For all integers $k\geq 0$, $\bx^\top M^k\bx=\by^\top M^k\by$.
\item There exists an orthogonal matrix $Q$ that commutes with $M$ such that $Q^2=I$ and $Q\bx=\by$.
\end{enumerate}
\end{lemma}

\begin{lemma}
\label{Prop:cosp}
If $\bx$ and $\by$ are strongly cospectral, then $\bx^\top M^k\bx=\by^\top M^k\by$ for all integers $k\geq 0$.
\end{lemma}

\begin{proof}
Since $E_j\bx=\pm E_j\by$, we get $\bx^\top E_j\bx=\by^\top E_j\by$. Invoking Lemma \ref{prop:}(2) yields the desired result.
\end{proof} 

The following result will prove useful in the latter sections.

\begin{theorem}
\label{Prop:sc}
Let $\bx,\by\in \R^n\backslash\{\zero\}$ and suppose
$\bx=\sum_{j\in \sigma_{\bx}(M)}\bu_j$, where each $\bu_j$ is a real eigenvector associated with an eigenvalue $\lambda_j\in \sigma_{\bx}(M)$. The following are equivalent.
\begin{enumerate}
\item The vectors $\bx$ and $\by$ are strongly cospectral.
\item There exists an orthogonal matrix $Q$ that is a polynomial in $M$ such that $Q^2=I$ and $Q\bx=\by$.
\item For some nonempty sets $\sigma_1$ and $\sigma_2$ that partition $\sigma_{\bx}(M)$, we have
\begin{equation*}
\bx=\sum_{\lambda_j\in \sigma_1}\bu_j+\sum_{\lambda_j\in \sigma_2}\bu_j\quad \text{and}\quad \by=\sum_{\lambda_j\in \sigma_1}\bu_j-\sum_{\lambda_j\in \sigma_2}\bu_j.
\end{equation*}
In this case, $\sigma_{\bx,\by}^{+}(M)=\sigma_1$ and $\sigma_{\bx,\by}^{-}(M)=\sigma_2$.
\end{enumerate}
\end{theorem}

\begin{proof}
The proof of the equivalence of (1) and (2) is analogous to that of \cite[Theorem 10.2]{GodsilSmith2024}. To prove (1) implies (3), suppose $\bx$ and $\by$ are strongly cospectral. Let $\lambda_k\in\sigma_{\bx}(M)=\sigma_{\by}(M)$. By assumption, $E_k\bx=\bu_k$. Thus, $E_k\bx=E_k\by$ if and only if $E_k\by=\bu_k$, while $E_k\bx=-E_k\by$ if and only if $E_k\by=-\bu_k$. Consequently,
\begin{center}
$\displaystyle\by=\sum_{\lambda_j\in\sigma_{\by}(M)}E_j\by=\sum_{\lambda_j\in \sigma_{\bx,\by}^{+}(M)}\bu_j-\sum_{\lambda_j\in \sigma_{\bx,\by}^{-}(M)}\bu_j$.
\end{center}
The converse is straightforward.
\end{proof}

\begin{remark}
\label{numberofpartitions}
Suppose $\bx\in \R^n\backslash\{\zero\}$ is not a fixed state and let $m=|\sigma_{\bx}(M)|$. There are $2^{m-1}-1$ ways to partition $\sigma_{\bx}(M)$ into two subsets. By Theorem \ref{Prop:sc}(3), we get $2^{m-1}-1$ vectors $\by$ that are strongly cospectral with $\bx$ in $G$. 
\end{remark}

The following result generalizes Corollary~6.4 of \cite{GodsilSmith2024} to real pure states.

\begin{lemma}
\label{Lem:autfix}
If $\bx$ and $\by$ are strongly cospectral, then any automorphism of $G$ that fixes $\by$ also fixes $\bx$.
\end{lemma}

\begin{proof}
Let $P$ be a permutation matrix representing an automorphism of $G$ that fixes $\by$. That is, $P\by=\by$ and $PE_j=E_jP$ for each $\lambda_j\in\sigma_{\bx}(M)$. Thus, $E_j\bx=\pm E_j\by=\pm P^\top E_jP\by=\pm P^\top E_j\by=P^\top E_j\bx$.
Since the $E_j$s sum to identity, the above equation yields $P\bx=\bx$.
\end{proof}

\section{Perfect state transfer}
\label{Sec:pst}

Recall that a fixed state cannot be involved in strong cospectrality. Hence, we restrict our discussion of PST to vectors that are not eigenvectors for $M$.

If $\bx$ and $\by$ are vertex states (respectively, pair states, plus states and $s$-pair states), then we sometimes use the term vertex PST (respectively, pair PST, plus PST and $s$-pair PST) in lieu of PST between $\bx$ and $\by$. 

\begin{lemma}
\label{Lem:pstnec}
Let $\bx,\by\in \R^n\backslash\{\zero\}$ with $\by\neq\pm \bx$. If perfect state transfer occurs between $\bx$ and $\by$ in $G$ at time $\tau>0$, then the following hold.
\begin{enumerate}
    \item The vectors $\bx$ and $\by$ are strongly cospectral, and $\bx$ and $\by$ are periodic at time $2\tau$.
    \item The minimum PST time between $\bx$ and $\by$ is $\frac{\rho}{2}$, where $\rho$ is given in Lemma \ref{Lem:minperiod}.
    \item If perfect state transfer also occurs between $\bx$ and $\bz$ in $G$, then $\by=\bz$.
\end{enumerate}
\end{lemma}

\begin{proof}
This is immediate from Lemma 2.3, Lemma 5.2 and Corollary 5.3 of \cite{godsil2017real}, respectively.
\end{proof}

We now provide a characterization of PST between real pure states. Our result applies even if $\phi(M,t)\notin\mathbb{Z}[x]$ (that is, even if $\sigma_{\bx}(M)$ is not closed under algebraic conjugates). We denote the largest power of two that divides an integer $a$ by $\nu_2(a)$. We adopt the convention that $\nu_2(0)=+\infty$.

\begin{theorem}
\label{Thm:PSTchar}
Let $G$ be a graph on $n$ vertices and let $\bx,\by \in \R^n\backslash\{\zero\}$ with $\by \ne \pm \bx$.
\begin{enumerate}
\item If $|\sigma_{\bx}(M)|=2$, then $\bx$ and $\by$ admit perfect state transfer if and only if they are strongly cospectral.
\item Let $\sigma_{\bx}(M)=\{\lambda_1,\ldots,\lambda_m\}$ for some $m\geq 3$. The vectors $\bx$ and $\by$ admit perfect state transfer if and only if $\bx$ and $\by$ are strongly cospectral, $\frac{\lambda_1-\lambda_j}{\lambda_1-\lambda_2}=\frac{p_j}{q_j}$ for each $j\geq 3$, where $p_j$ and $q_j$ are integers such that $\operatorname{gcd}(p_j,q_j)=1$, and one of the following conditions holds.
\begin{enumerate}
\item If $\lambda_1,\lambda_2\in\sigma_{\bx,\by}^+(M)$, then $$\nu_2(q_{\ell})=\nu_2(q_{k})>\nu_2(q_h)$$ for all $\lambda_{\ell},\lambda_k\in \sigma_{\bx,\by}^-(M)$ and $\lambda_h\in \sigma_{\bx,\by}^+(M)\backslash\{\lambda_1,\lambda_2\}$, where each $\nu_2(q_h)$ above is absent whenever $|\sigma_{\bx,\by}^+(M)|=2$.
\item If $\lambda_1\in\sigma_{\bx,\by}^+(M)$ and $\lambda_2\in\sigma_{\bx,\by}^-(M)$, then each $q_j$ is odd, and $p_h$ is even if and only if $\lambda_h\in\sigma_{\bx,\by}^+(M)$.
\end{enumerate}
\end{enumerate}
Moreover, the minimum PST time is $\frac{\rho}{2}$, where $\rho$ is given in Lemma \ref{Lem:minperiod}.
\end{theorem}

\begin{proof}
The last statement follows from Lemma \ref{Lem:pstnec}(2). Thus, we may assume that the minimum PST time is $\tau=\frac{\rho}{2}$, where $\rho$ is given in Lemma \ref{Lem:minperiod}. Without loss of generality, let $\lambda_1\in\sigma_{\bx,\by}^+(M)$. For all $\lambda_h\in\sigma_{\bx,\by}^+(M)$ and $\lambda_{\ell}\in\sigma_{\bx,\by}^-(M)$, equation (\ref{Eq:1}) implies that
\begin{equation}
\label{Eq:pst1}
e^{i\tau(\lambda_1-\lambda_h)}=-e^{i\tau(\lambda_1-\lambda_{\ell})}.
\end{equation}

By Lemma \ref{Lem:pstnec}(1), it suffices to prove the converse of (1). Let $|\sigma_{\bx}(M)|=2$, and $\bx$ and $\by$ be strongly cospectral. Since $\lambda_1\in\sigma_{\bx,\by}^+(M)$ and $\sigma_{\bx,\by}^-(M)\neq \varnothing$, we have $\sigma_{\bx,\by}^-(M)=\{\lambda_2\}$. From equation (\ref{Eq:pst1}), we get $e^{i\tau(\lambda_1-\lambda_2)}=-1$, and so PST occurs between $\bx$ and $\by$ at $\tau=\frac{\pi}{\lambda_1-\lambda_2}$.

We now prove (2). Suppose PST occurs between $\bx$ and $\by$. By Lemma \ref{Lem:pstnec}(1), $\bx$ and $\by$ are strongly cospectral and periodic. By Theorem \ref{Thm:rc}, the latter statement is equivalent to $\frac{\lambda_1-\lambda_j}{\lambda_1-\lambda_2}=\frac{p_j}{q_j}$ for each $j\geq 3$, where $p_j$ and $q_j$ are integers such that $\operatorname{gcd}(p_j,q_j)=1$. In this case, $\tau=\frac{\pi q}{\lambda_1-\lambda_2}$, where $q=\operatorname{lcm}(q_3,\ldots,q_{m})$.
To prove (2a), let 
$\lambda_2\in \sigma_{\bx,\by}^+(M)$. Equation (\ref{Eq:pst1}) holds if and only if
\begin{equation}
\label{Eq:pst2}
    \tau\left(\lambda_1-\lambda_{\ell}\right)=\pi q\left(\frac{\lambda_1-\lambda_{\ell}}{\lambda_1-\lambda_2}\right)=\frac{\pi q p_{\ell}}{q_{\ell}}\equiv
\begin{cases}
    0\ \text{(mod $2\pi$)},& \text{if } \lambda_{\ell}\in \sigma_{\bx,\by}^+(M)\\
    \pi\ \text{(mod $2\pi$)},& \text{if } \lambda_{\ell}\in \sigma_{\bx,\by}^-(M).
\end{cases}
\end{equation}
Since $\lambda_2\in \sigma_{\bx,\by}^+(M)$, we have $\tau\left(\lambda_1-\lambda_{2}\right)=\pi q\left(\frac{\lambda_1-\lambda_{2}}{\lambda_1-\lambda_2}\right)=\pi q\equiv 0$ (mod $2\pi$), and so $q$ is even. Now, for equation (\ref{Eq:pst2}) to hold for each $\lambda_{\ell}\in \sigma_{\bx,\by}^-(M)$, we need each $\frac{qp_{\ell}}{q_{\ell}}$ to be odd in which case $\nu_2(q)=\nu_2(q_{\ell})$. Equivalently, the $\nu_2(q_{\ell})$'s must all be equal. On the other hand, for (\ref{Eq:pst2}) to hold for each $\lambda_h\in \sigma_{\bx,\by}^+(M)$, we need each $\frac{qp_{h}}{q_{h}}$ to be even. That is, we must have $\nu_2(q)=\nu_2(q_{\ell})>\nu_2(q_{h})$ for each $\lambda_h\in \sigma_{\bx,\by}^+(M)$. Thus (2a) holds, and (2b) can be proven similarly. The converse of (2) is straightforward.
\end{proof}

\begin{remark}
Suppose $G$ has at least three vertices. If $\bx=\be_u$ and $\by=\be_v$ are strongly cospectral in $G$, then $|\sigma_{\bx}(M)|\geq 3$ by \cite[Theorem 3.4]{MonterdeELA} and so Theorem \ref{Thm:PSTchar}(1) does not apply to vertex states.
\end{remark}

Combining Theorem \ref{Prop:sc}(3) and Theorem \ref{Thm:PSTchar}(1) yields the following result.

\begin{corollary}
\label{Cor:size2}
Let $\bx,\by \in \R^n\backslash\{\zero\}$. If $\sigma_{\bx}(M)=\{\lambda_1,\lambda_2\}$, then $\bx$ and $\by$ admit perfect state transfer if and only if $\bx=\bu_1+\bu_2$ and $\by=\bu_1-\bu_2$, where $\bu_1,\bu_2$ are real eigenvectors associated with $\lambda_1,\lambda_2$.
\end{corollary}

\begin{theorem}
\label{Thm:allgraphs}
All connected graphs on $n\geq 2$ vertices admit perfect state transfer between real pure states.
\end{theorem}

\begin{proof}
By assumption, $M\in\{A,L\}$ has at least two distinct eigenvalues, say $\lambda_1,\lambda_2$. By Corollary \ref{Cor:size2}, there is PST  between $\bx=\bu_1+\bu_2$ and $\by=\bu_1-\bu_2$, where $\bu_1,\bu_2$ are real eigenvectors associated with $\lambda_1,\lambda_2$.
\end{proof}

\begin{remark}
The real pure states admitting PST in Theorem \ref{Thm:allgraphs} have eigenvalue supports of size two. For those with eigenvalue supports of size at least three, periodicity is required by Theorem \ref{Thm:PSTchar}(2) to achieve PST. However, there are graphs for which periodicity does not hold for such real pure states (e.g., conference graphs on $n$ vertices, where $n$ is not a perfect square).
Thus, the conclusion of Theorem \ref{Thm:allgraphs} can fail for some connected graphs whenever the real pure states in question have eigenvalue supports of size at least three.
\end{remark}

Combining Theorem \ref{Thm:PSTchar}(2) and Theorem \ref{Thm:rc} yields an extension of \cite[Theorem 2.4.4]{Coutinho2014}.

\begin{corollary}
\label{Cor:PSTcharcor}
Let $G$ be a graph on $n$ vertices and $\bx,\by \in\R^n\backslash\{\zero\}$ with $\by\neq \pm\bx$. If $|\sigma_{\bx}(M)|\geq 3$ and $\sigma_{\bx}(M)$ is closed under algebraic conjugates, then $\bx$ and $\by$ admit perfect state transfer if and only if all conditions below hold.
\begin{enumerate}
\item The vectors $\bx$ and $\by$ are strongly cospectral.
\item Each eigenvalue $\lambda_j\in\sigma_{\bx}(M)$ is of the form $\lambda_j=\frac{1}{2}(a+b_j\sqrt{\Delta})$, where $a$, $b_j$, and $\Delta$ are integers and either $\Delta=1$ or $\Delta>1$ is a square-free natural number.
\item Let $\lambda_1\in\sigma_{\bx,\by}^+(M)$. For all $\lambda_h\in\sigma_{\bx,\by}^+(M)$ and $\lambda_{\ell},\lambda_k\in\sigma_{\bx,\by}^-(M)$, we have
\begin{center}
$\nu_2\left(\frac{\lambda_1-\lambda_{h}}{\sqrt{\Delta}}\right)>\nu_2\left(\frac{\lambda_1-\lambda_{\ell}}{\sqrt{\Delta}}\right)=\nu_2\left(\frac{\lambda_1-\lambda_{k}}{\sqrt{\Delta}}\right).$
\end{center}
\end{enumerate}
In the case that there is PST between $\bx$ and $\by$, the minimum PST time is $\frac{\pi}{g\sqrt{\Delta}}$, where $g$ is given in Corollary \ref{Lem:minperiod1}.
\end{corollary}

As it turns out, a periodic real pure state always admits PST with another real pure state.

\begin{theorem}
\label{Thm:PSTany}
If $\bx\in\R^n\backslash\{\zero\}$ is periodic in $G$ at time $\tau$, then $\bx$ is involved in perfect state transfer in $G$ at $\frac{\tau}{2}$.
\end{theorem}

\begin{proof}
If $|\sigma_{\bx}(M)|=2$, then the same proof used in Theorem \ref{Thm:allgraphs} works. Now, let $|\sigma_{\bx}(M)|\geq 3$ and suppose $\bx=\sum_{j\in\sigma_{\bx}(M)}\bu_j$, where $\bu_j$ is a real eigenvector associated with the eigenvalue $\lambda_j\in\sigma_{\bx}(M)$. Since $\bx$ is periodic, Theorem \ref{Thm:rc} implies that we may write $\frac{\lambda_1-\lambda_j}{\lambda_1-\lambda_2}=\frac{p_j}{q_j}$ for each $j\geq 3$, where $p_j,q_j$ are integers such that $\operatorname{gcd}(p_j,q_j)=1$. If each $p_j$ and $q_j$ are odd, then define $\by:=\bu_1-\sum_{j\neq 1}\bu_j$. Since $|\sigma_{\bx,\by}^+(M)|=1$, invoking Theorem \ref{Thm:PSTchar}(2b) yields PST between $\bx$ and $\by$. If some $p_j$'s are even and each $q_j$ is odd, then let $\sigma_1=\{\lambda_j:p_j\ \text{is even}\}$, $\sigma_2=\sigma_{\bx}(M)\backslash\sigma_1$, and define $\by:=\sum_{j\in\sigma_1}\bu_j-\sum_{j\in\sigma_2}\bu_j$. Note that $\lambda_1\in\sigma_1$, $\lambda_2\in\sigma_2$ and $\sigma_{\bx,\by}^+(M)=\sigma_1$ has at least two elements. Since $\lambda_1\in \sigma_{\bx,\by}^+(M)$ and $\lambda_2\in \sigma_{\bx,\by}^-(M)$, applying Theorem \ref{Thm:PSTchar}(2b) yields PST between $\bx$ and $\by$. Finally, if some $q_j$'s are even, then let $\sigma_2:=\{\lambda_{\ell}:\nu_2(q_{\ell}))=\eta\}$, where $\eta=\max_{j\geq 3} \nu_2(q_j)>0$. Define $\by=\sum_{j\in\sigma_1}\bu_j-\sum_{j\in\sigma_2}\bu_j$, where $\sigma_1=\sigma_{\bx}(M)\backslash \sigma_2$. Since $\sigma_{\bx,\by}^+(M)=\sigma_1$ has at least two elements and $\lambda_1,\lambda_2\in \sigma_{\bx,\by}^+(M)$, Theorem \ref{Thm:PSTchar}(2a) again yields PST between $\bx$ and $\by$. In all cases, $\by$ is unique by Lemma \ref{Lem:pstnec}(3). Moreover, $\by$ is real because the $\bu_j$s are all real.
\end{proof}

We close this section with the following result.

\begin{theorem}
\label{Thm:prescribedPST}
Let $\bx,\by\in\R^n\backslash\{\zero\}$ with $\by\neq\pm \bx$. For all $\tau>0$ and for all integers $m_1,m_2\geq 1$ such that $m_1+m_2\leq n$, there exists a real symmetric matrix $M$ such that perfect state transfer occurs between $\bx$ and $\by$ relative to $M$ at time $\tau$, $|\sigma_{\bx,\by}^+(M)|=m_1$ and $|\sigma_{\bx,\by}^-(M)|=m_2$.
\end{theorem}

\begin{proof}
Let $\bx,\by\in\R^n$ with $\by\neq\pm \bx$. Fix $\tau>0$ and fix integers $m_1,m_2\geq 1$ such that $m_1+m_2\leq n$. Since $\|\bx\|=\|\by\|$, it follows that $\bx+\by$ and $\bx-\by$ are orthogonal vectors. Since
\begin{center}
$\{\bu_1,\bu_2\}:=\left\{\frac{1}{\sqrt{m_1}}\left[\begin{array}{ccc} \ones_{m_1} \\ \zero_{m_2} \\ \zero_{n-(m_1+m_2)} \end{array} \right],\frac{1}{\sqrt{m_2}}\left[\begin{array}{ccc} \zero_{m_1} \\ \ones_{m_2} \\ \zero_{n-(m_1+m_2)} \end{array} \right]\right\}\quad $ and $\quad \{\bv_1,\bv_2\}:=\left\{\frac{\bx+\by}{\|\bx+\by\|},\frac{\bx-\by}{\|\bx-\by\|}\right\}$    
\end{center}
are orthonormal sets, we may extend them to orthonormal bases $\mathcal{U}=\{\bu_1,\ldots,\bu_n\}$ and $\mathcal{V}=\{\bv_1,\ldots,\bv_n\}$ for $\R^n$, respectively. Then there exists an orthogonal matrix $Q$ such that $Q\bu_j=\bv_j$ for each $j$. If we write $Q=[\bw_1,\ldots,\bw_n]$, where $\bw_1,\ldots,\bw_n$ are columns of $Q$, then $Q\bu_1=\bv_1$ and $Q\bu_2=\bv_2$ are equivalent to
\begin{equation*}
\bx+\by=\displaystyle\frac{\|\bx+\by\|}{\sqrt{m_1}}\left(\sum_{j=1}^{m_1}\bw_j\right)\quad \text{and} \quad \bx-\by=\displaystyle\frac{\|\bx-\by\|}{\sqrt{m_2}}\left(\sum_{j=m_1+1}^{m_2}\bw_j\right).
\end{equation*}
Thus, we obtain
\begin{equation}
\label{Eq:sc}
\bx=\displaystyle\left(\sum_{j=1}^{m_1}\displaystyle\frac{\|\bx+\by\|}{2\sqrt{m_1}}\bw_j\right)+\left(\sum_{j=m_1+1}^{m_2}\displaystyle\frac{\|\bx-\by\|}{2\sqrt{m_2}}\bw_j\right)\quad \text{and} \quad \by=\displaystyle\left(\sum_{j=1}^{m_1}\displaystyle\frac{\|\bx+\by\|}{2\sqrt{m_1}}\bw_j\right)-\left(\sum_{j=m_1+1}^{m_2}\displaystyle\frac{\|\bx-\by\|}{2\sqrt{m_2}}\bw_j\right)
\end{equation}
Now, let 
$\{\theta_1,\ldots,\theta_{n}\}$ be a set of $n$ distinct real numbers and consider the real symmetric matrix
\begin{equation*}
M=\sum_{j=1}^{n}\theta_j\bw_j\bw_j^\top
\end{equation*}
Since equation (\ref{Eq:sc}) holds, Theorem \ref{Prop:sc}(3) implies that $\bx$ and $\by$ are strongly cospectral relative to the matrix $M$ with $\sigma_{\bx,\by}^+(M)=\{\theta_1,\ldots,\theta_{m_1}\}$ and $\sigma_{\bx,\by}^-(M)=\{\theta_{m_1+1},\ldots,\theta_{m_1+m_2}\}$. We proceed with cases.

\noindent \textbf{Case 1.} Suppose $m_1+m_2=2$. In this case, choose $\theta_1,\theta_2$ such that $\theta_1-\theta_2=\frac{\pi}{\tau}$. Since $|\sigma_{\bx}(M)|=2$, invoking Theorem \ref{Thm:PSTchar}(1) yields PST between $\bx$ and $\by$ relative to $M$ at time $\frac{\pi}{\theta_1-\theta_2}=\frac{\pi}{\pi/\tau}=\tau$.

\noindent \textbf{Case 2.} Suppose $m_1+m_2\geq 3$. In this case, choose $\theta_1,\ldots,\theta_{m_1+m_2}$ such that $\theta_1-\theta_j=\pi b_j/(g\tau)$ where $b_j$ is even for all $j\in\{2,\ldots,m_1\}$, $b_j$ is odd otherwise, and $g=\operatorname{gcd}(b_2,\ldots,b_{m_1+m_2})$. Observe that
\begin{equation}
\label{Eq:PST}
\frac{\theta_1-\theta_j}{\theta_1-\theta_2}=\frac{b_j/g_j}{b_2/g_j}:=\frac{p_j}{q_j},
\end{equation}
where each $g_j=\operatorname{gcd}(b_j,b_2)$.
Thus, $q=\operatorname{lcm}(q_3,\ldots,q_{m_1+m_2})=\operatorname{lcm}(\frac{b_2}{g_3},\ldots,\frac{b_2}{g_{m_1+m_2}})=\frac{b_2}{g}$. We have two cases. First, if $m_1=1$, then each $p_j$ and $q_j$ is odd. In this case, invoking Theorem \ref{Thm:PSTchar}(2b) yields PST between $\bx$ and $\by$ relative to $M$ at time $\frac{\pi q}{\theta_1-\theta_2}=\frac{\pi b_2/g}{\pi b_2/ (g\tau)}=\tau$. Now, if $m_1\geq 2$, then $b_{\ell}$ is odd for each $\theta_{\ell}\in\sigma_{\bx,\by}^-(M)$, and so $\nu_2(g_{\ell})=\nu_2(\operatorname{gcd}(b_{\ell},b_2))=0$. Consequently,  $\nu_2(q_{\ell})=\nu_2(\frac{b_2}{g_{\ell}})=\nu_2(b_2)$ for each $\theta_{\ell}\in\sigma_{\bx,\by}^-(M)$. Moreover, $b_{h}$ is even for each $\theta_{h}\in\sigma_{\bx,\by}^+(M)$, and so $\nu_2(g_{h})=\nu_2(\operatorname{gcd}(b_{h},b_2))\geq 1$. Consequently,  $\nu_2(q_{h})=\nu_2(\frac{b_2}{g_{h}})<\nu_2(b_2)$ for each $\theta_{h}\in\sigma_{\bx,\by}^+(M)$. Equivalently, $$\nu_2(q_{\ell})=\nu_2(q_{k})>\nu_2(q_h)$$ for all $\theta_{\ell},\theta_k\in \sigma_{\bx,\by}^-(M)$ and $\theta_h\in \sigma_{\bx,\by}^+(M)\backslash\{\theta_1,\theta_2\}$, where each $\nu_2(q_h)$ above is absent whenever $|\sigma_{\bx,\by}^+(M)|=2$. Finally, invoking Theorem \ref{Thm:PSTchar}(2a) yields PST between $\bx$ and $\by$ at time $\tau$.

Combining the above two cases yields the desired result.
\end{proof}

\section{Complete graphs and cycles}\label{Sec:cycle}

Here, we characterize real pure states that exhibit PST in complete graphs and cycles. Since these graphs are regular, it suffices to investigate the quantum walk relative to $A$.

\begin{theorem}
\label{Thm:pstkn}
The vectors $\bx,\by\in\R^n\backslash\{\zero\}$ with $\by\neq \pm\bx$ admit perfect state transfer in $K_n$ if and only if $\by=\bx-\frac{2(\ones ^\top\bx)}{n}\ones$. In this case, the minimum PST time is $\frac{\pi}{n}$.
\end{theorem}

\begin{proof}
For $K_n$, $U_A(t)=e^{-it}\left((e^{itn}-1)\frac{1}{n}J+I\right).$ Thus, $e^{it}U_A(t)\bx=\bx+\left(\frac{(e^{itn}-1)\ones^\top\bx}{n}\right)\ones.$ If $\bx=\frac{1}{n}\ones$ or $\bx\perp \ones$, then $\bx$ is fixed. Otherwise, $\bx$ is periodic with $\rho=\frac{2\pi}{n}$. Invoking Corollary \ref{Cor:PSTcharcor}, $\frac{\pi}{n}$ is the minimum PST time from $\bx$, in which case we obtain $e^{i\pi/n}U_A(\tau)\bx=\bx-\frac{2(\ones ^\top\bx)}{n}\ones$.
\end{proof}

\begin{example}
\label{Ex:Kn}
In $K_2$, if $V(K_2)=\{u,v\}$, then PST happens between $\be_u+s\be_v$ and $\be_v+s\be_u$ at $\frac{\pi}{2}$ for all $s\neq \pm 1$.
In $K_3$, if $V(K_3)=\{u,v,w\}$, then PST happens between the pairs $\{\be_u+2\be_w,\be_u+2\be_v\}$ and $\{\be_u+\frac{1}{2}\be_w,\be_v+\frac{1}{2}\be_w\}$ both at $\frac{\pi}{3}$. In $K_4$, if $V(K_4)=\{u,v,w,x\}$, then PST happens between $\be_u+\be_w$ and $\be_v+\be_x$ at $\frac{\pi}{4}$.
\end{example}

If
$\bx=\be_u+s\be_w$, then $\bx-\frac{2(\ones ^\top\bx)}{n}\ones 
=\be_u+s\be_w-\frac{2(1+s)}{n}\ones.$ Hence, if $s=-1$ then $\bx$ is fixed, while if $s\neq -1$, and $n \ge 5$ then $\bx-\frac{2(\ones ^\top\bx)}{n}\ones $ is not an $s$-pair state. Thus, $s$-pair PST does not occur in $K_n$ for all $n\geq 5$. Together with Example \ref{Ex:Kn}, we have the following result.

\begin{corollary}
Perfect state transfer between $s$-pair states occurs in $K_n$ if and only if $n\in\{2,3,4\}$.
\end{corollary}

For cycles $C_n$, we adopt the convention that $V(C_n)=\mathbb{Z}_n$ where vertices $j, k$ are adjacent if and only if $|k-j| \equiv 1 \mod n.$ The eigenvalues and eigenvectors of $C_n$ are well-known, see \cite[Section 1.4.3]{brouwer2011spectra}. For our purposes, we provide normalized eigenvectors for $C_n$ in the following lemma.

\begin{lemma}
\label{Lem:cneval}
The adjacency eigenvalues of $C_n$ are $\lambda_j=2\cos(2j\pi/n)$, where $0\leq j\leq \lfloor\frac{n}{2}\rfloor$. The associated eigenvector for $\lambda_0=2$ is $\bv_0=\frac{1}{\sqrt{n}}\ones$, while the associated eigenvector for $\lambda_{\frac{n}{2}}=-2$ whenever $n$ is even is $\bv_{\frac{n}{2}}=\frac{1}{\sqrt{n}}[1,-1,1,-1,\ldots,1,-1]^\top$. For $1\leq j<\frac{n}{2}$, we have the following associated eigenvectors for $\lambda_j$:
\begin{center}
$\bv_j=\sqrt{\frac{2}{n}}\left[1\;\; \cos\left(\frac{2j\pi}{n}\right)\;\; \cos\left(\frac{4j\pi}{n}\right)\;\; \cdots\;\; \cos\left(\frac{2j(n-1)\pi}{n}\right)\right]^\top$
\end{center}
and
\begin{center}
$\bv_{n-j}=\sqrt{\frac{2}{n}}\left[0\;\; \sin\left(\frac{2j\pi}{n}\right)\;\; \sin\left(\frac{4j\pi}{n}\right)\;\; \cdots\;\; \sin\left(\frac{2j(n-1)\pi}{n}\right)\right]^\top$.
\end{center}
Moreover, $\{\bv_0,\ldots,\bv_{n-1}\}$ is an orthonormal basis for $\R^n$.
\end{lemma}

We are now ready to characterize real pure states that admit PST in cycles. Since Corollary \ref{Cor:size2} takes care of the case $|\sigma_{\bx}(A)|=2$, we only focus on the case when $|\sigma_{\bx}(A)|\geq 3$.

\begin{theorem}
\label{Thm:pstcycles}
Let $n\geq 3$ and $\bx,\by\in\R^n$. Suppose $|\sigma_{\bx}(A)|\geq 3$, and $\sigma_{\bx}(A)$ is closed under algebraic conjugates. Then $C_n$ admits perfect state transfer between $\bx$ and $\by$ if and only if either:
\begin{enumerate}

\item $n=2m$, $\bx=a\bv_0+b(\alpha_1\bv_{\frac{n}{6}}+\alpha_2\bv_{\frac{5n}{6}})+c(\beta_1\bv_{\frac{n}{4}}+\beta_2\bv_{\frac{3n}{4}})+d(\gamma_1\bv_{\frac{n}{3}}+\gamma_2\bv_{\frac{2n}{3}})+e\bv_{\frac{n}{2}}$    
and $\by=a\bv_0-b\alpha_1\bv_{\frac{n}{6}}-b\alpha_2\bv_{\frac{5n}{6}}+c\beta_1\bv_{\frac{n}{4}}+c\beta_2\bv_{\frac{3n}{4}}-d\gamma_1\bv_{\frac{n}{3}}-d\gamma_2\bv_{\frac{2n}{3}}+e\bv_{\frac{n}{2}}$,
and either
\begin{enumerate}
\item If $c=0$, then $m\equiv 0$ (mod 3). In this case $\sigma_{\bx}(M)\subseteq\{\pm 1,\pm 2\}$.
\item Else, $m\equiv 0$ (mod 6). In this case, $0\in \sigma_{\bx}(M)$, $\sigma_{\bx}(M)\subseteq\{0,\pm 1,\pm 2\}$ and $\sigma_{\bx}(M)\neq \{0,\pm 2\}$.
\end{enumerate}

\item $n=4m$, $\bx=a\bv_0+b(\beta_1\bv_{m}+\beta_2\bv_{3m})+c\bv_{2m}$    
and $\by=-a\bv_0+b(\beta_1\bv_{m}+\beta_2\bv_{3m})-c\bv_{2m}$. In this case $\sigma_{\bx}(M)=\{ 0,\pm 2\}$.

\item $n=12m$, $\bx=a(\alpha_1\bv_{3m}+\alpha_2\bv_{9m})+b(\beta_1\bv_{m}+\beta_2\bv_{11m})+c(\gamma_1\bv_{5m}+\gamma_2\bv_{7m})$ and  $\by=a(\alpha_1\bv_{3m}+\alpha_2\bv_{9m})-b(\beta_1\bv_{m}+\beta_2\bv_{11m})-c(\gamma_1\bv_{5m}+\gamma_2\bv_{7m})$. In this case, $\sigma_{\bx}(M)=\{ 0,\pm \sqrt{3}\}$.

\item $n=8m$, $\bx=a(\alpha_1\bv_{2m}+\alpha_2\bv_{6m})+b(\beta_1\bv_{m}+\beta_2\bv_{7m})+c(\gamma_1\bv_{3m}+\gamma_2\bv_{5m})$ and $\by=a(\alpha_1\bv_{2m}+\alpha_2\bv_{6m})-b(\beta_1\bv_{m}+\beta_2\bv_{7m})-c(\gamma_1\bv_{3m}+\gamma_2\bv_{5m})$. In this case $\sigma_{\bx}(M)=\{ 0,\pm \sqrt{2}\}$.
\end{enumerate}
In all cases above, $a,b,c,d,e\in\R$ are such that $a^2+b^2+c^2=\|\bx\|^2$ in (2)-(4) with $a,b,c\neq 0$, and $a^2+b^2+c^2+d^2+e^2=\|\bx\|^2$ otherwise. Moreover, $(\alpha_1,\alpha_2),(\beta_1,\beta_2),(\gamma_1,\gamma_2)\in \R^2\backslash\{(0,0)\}$ such that $\alpha_1^2+\alpha_2^2=\beta_1^2+\beta_2^2=\gamma_1^2+\gamma_2^2=1$. The minimum PST time in (1)-(4) is $\pi$, $\frac{\pi}{2}$, $\frac{\pi}{\sqrt{3}}$ and $\frac{\pi}{\sqrt{2}}$, respectively.
\end{theorem}

\begin{proof}
Let $|\sigma_{\bx}(A)|\geq 3$ and $\sigma_{\bx}(A)$ be closed under taking algebraic conjugates. If $\bx$ is periodic, then the elements in $\sigma_{\bx}(A)$ differ by at least one by Corollary \ref{Cor:specgap}. Since $|\lambda_j|\leq 2$, we get $|\sigma_{\bx}(A)|\leq 5$. By Theorem \ref{Thm:rc}, we have two cases.
\vspace{0.05in}

\noindent \textbf{Case 1.} Let $\sigma_{\bx}(A)\subseteq\mathbb{Z}$ so that $\sigma_{\bx}(A)\subseteq \{0,\pm 1,\pm 2\}$. Invoking Proposition \ref{Prop:S}, we may write $\bx=a\bv_0+b\bu+c\bw+d\bz+e\bv_{\frac{n}{2}}$, where $a,b,c,d,e\in\R$ with $a^2+b^2+c^2+d^2+e^2=1$, and $\bu,\bw,\bz$ are eigenvectors associated with $\lambda_{\frac{n}{6}}=1$, $\lambda_{\frac{n}{4}}=0$, $\lambda_{\frac{n}{3}}=-1$, respectively. The latter implies that we may write $\bu=\alpha_1\bv_{\frac{n}{6}}+\alpha_2\bv_{\frac{5n}{6}}$, $\bw=\beta_1\bv_{\frac{n}{4}}+\beta_2\bv_{\frac{3n}{4}}$, and $\bz=\gamma_1\bv_{\frac{n}{3}}+\gamma_2\bv_{\frac{2n}{3}}$, where $(\alpha_1,\alpha_2),(\beta_1,\beta_2),(\gamma_1,\gamma_2)\in\R^2\backslash\{(0,0)\}$ satisfying $\alpha_1^2+\alpha_2^2=\beta_1^2+\beta_2^2=\gamma_1^2+\gamma_2^2=1$. From this, it is clear that $n=2m$ for some integer $m$. We proceed with two subcases.
\begin{itemize}
    \item Suppose $b\neq 0$ or $d\neq 0$. If $c=0$, then $\frac{n}{6}$ or $\frac{n}{3}$ is an integer, which implies that $m\equiv 0$ (mod 3). Otherwise, $m\equiv 0$ (mod 6). This proves (1).
    \item Suppose $b=0$ and $d=0$. Then we may rewrite $\bx=a\bv_0+b'\bw+c'\bv_{\frac{n}{2}}$, where $\bw=\beta_1\bv_{\frac{n}{4}}+\beta_2\bv_{\frac{3n}{4}}$. Thus, $n\equiv 0$ (mod 4). This proves (2).
\end{itemize}

\noindent \textbf{Case 2.} Let $\sigma_{\bx}(A)\subseteq\{\frac{a}{2},\frac{1}{2}(a\pm b_1\sqrt{\Delta}),\frac{1}{2}(a\pm b_2\sqrt{\Delta})\}$, where $a,b_1,b_2,\Delta$ are integers such that $b_1>b_2>0$ and $\Delta>1$ is square-free. Note that $\frac{a}{2}$ is an integer as it is a root of a polynomial with integer coefficients. Since $|\frac{a}{2}|\leq 2$, we get $a\in\{0,\pm 2,\pm 4\}$. If $a=2$, then $|\frac{1}{2}(2\pm b_j\sqrt{\Delta})|\leq 2$ implies that $b_j=1$ and $\Delta\in\{2,3\}$. However, $\frac{1}{2}(2\pm \sqrt{\Delta})$ are not quadratic integers whenever $\Delta\in\{2,3\}$. Thus, $a\neq 2$ and similarly, $a\neq -2$. Now, if $a=0$, then $|\frac{1}{2}b_j\sqrt{\Delta}|\leq 2$, where $\frac{1}{2}b_j$ is an integer, so we get $\frac{1}{2}b_j=1$ and $\Delta\in\{2,3\}$. Since $b_1>b_2$ and $\sigma_{\bx}(A)$ is closed under algebraic conjugates, $|\sigma_{\bx}(A)|\in\{4,5\}$ cannot hapen. Thus, $|\sigma_{\bx}(A)|=3$. We have the following subcases.

\begin{itemize}

\item Let $\Delta=3$. Then $\sigma_{\bx}(A)=\{0,\pm \sqrt{3}\}$. By Lemma \ref{Lem:cneval}, $n=12m$ and $j\in\{\frac{n}{4},\frac{n}{12},\frac{5n}{12}\}$. Therefore, $\bx=a\bz+b\bu+c\bw$, where $\bz=\alpha_1\bv_{\frac{n}{4}}+\alpha_2\bv_{\frac{3n}{4}}$, $\bu=\beta_1\bv_{\frac{n}{12}}+\beta_2\bv_{\frac{11n}{12}}$ and $\bw=\gamma_1\bv_{\frac{5n}{12}}+\gamma_2\bv_{\frac{7n}{12}}$.

\item Let $\Delta=2$. Then $\sigma_{\bx}(A)=\{0,\pm \sqrt{2}\}$. A similar argument yields $n=8m$ and $\bx=a\bz+b\bu+c\bw$, where $\bz=\alpha_1\bv_{\frac{n}{4}}+\alpha_2\bv_{\frac{3n}{4}}$, $\bu=\beta_1\bv_{\frac{n}{8}}+\beta_2\bv_{\frac{7n}{8}}$ and $\bw=\gamma_1\bv_{\frac{3n}{8}}+\gamma_2\bv_{\frac{5n}{8}}$. 
\end{itemize}
Combining the two cases above yields (3) and (4). Finally, the Pythagorean relations involving $a,b,c,d,e$ and the pairs $(\alpha_1,\alpha_2),(\beta_1,\beta_2),(\gamma_1,\gamma_2)$ follows from the fact that $\{\bv_j\}$ is an orthonormal basis for $\R^n$.
\end{proof}

It is known that cycles admit $s$-pair state transfer if and only if $n\in\{4,6,8\}$ \cite{kim2024generalization}. If we consider real pure states in general, then Theorem \ref{Thm:pstcycles} implies that there are infinite families of cycles that admit PST.

We close this section with the following example.

\begin{example}
\label{Ex:c12}
Let $\bx=\sum_{j=0}^{m-1}\be_{4j}$ and $\by=\sum_{j=0}^{m-1}\be_{4j+2}$ in $C_{4m}$. We may write $\bx=\frac{\sqrt{n}}{4}(\bv_0+\sqrt{2}\bv_m+\bv_{2m})$ and $\by=\frac{\sqrt{n}}{4}(\bv_0-\sqrt{2}\bv_m+\bv_{2m})$.
Invoking Theorem \ref{Thm:pstcycles}(2) with $a=c=\frac{\sqrt{n}}{4}$, $\beta_1=1$ (so that $\beta_2=0$) and $b=\frac{\sqrt{n}}{2\sqrt{2}}$ yields PST between $\bx$ and $\by$ at time $\frac{\pi}{2}$. In particular, if $m=2$, then we recover the fact that $C_8$ admits plus PST between $\bx=\be_0+\be_4$ and $\by=\be_2+\be_6$ \cite{kim2024generalization}. If $m=3$, then $C_{12}$ admits PST between $\bx=\be_0+\be_4+\be_8$ and $\by=\be_2+\be_6+\be_{10}$. This complements the fact that $C_{12}$ does not admit $s$-pair PST.
\end{example}

\section{Paths}\label{Sec:paths}

In this section, we characterize adjacency and Laplacian PST between real pure states in paths. We adopt the convention that the vertices of $P_n$ are labelled so that vertices $j, k$ are adjacent whenever $|k-j| = 1.$ We start with the adjacency case. The adjacency eigenvalues and eigenvectors of $P_n$ are well-known, see \cite[Section 1.4.4]{brouwer2011spectra}. Again for our purposes, we provide normalized eigenvectors for $A(P_n)$ below.

\begin{lemma}
\label{Lem:pneval}
For $j\in\{1,\ldots, n\}$, the adjacency eigenvector of $P_n$ with eigenvalue $\mu_j=2\cos\left(\frac{j\pi}{n+1}\right)$ is 
\begin{center}
$\bz_j=\sqrt{\frac{2}{n+1}}\left[\sin\left(\frac{j\pi}{n+1}\right), \sin\left(\frac{2j\pi}{n+1}\right),\ldots, \sin\left(\frac{nj\pi}{n+1}\right)\right]^\top$.
\end{center}
Moreover, $\{\bz_1,\ldots,\bz_n\}$ forms an orthonormal basis for $\R^n$.
\end{lemma}

We now characterize real pure states that admit adjacency PST in paths. Again, since Corollary \ref{Cor:size2} takes care of the case $|\sigma_{\bx}(A)|=2$, we only focus on the case when $|\sigma_{\bx}(A)|\geq 3$.

\begin{theorem}
\label{Thm:pstpaths}
Let $n\geq 3$. Suppose $\bx,\by\in\R^n$ such that $|\sigma_{\bx}(A)|\geq 3$ and $\sigma_{\bx}(A)$ is closed under algebraic conjugates. $P_n$ admits perfect state transfer between $\bx$ and $\by$ if and only if either:
\begin{enumerate}

\item $n+1=6m$, and either
\begin{enumerate}
\item $\bx=a\bz_{3m}+b\bz_{2m}+c\bz_{4m}$ and $\by=a\bz_{3m}-b\bz_{2m}-c\bz_{4m}$,
\item $\bx=a\bz_{3m}+b\bz_{m}+c\bz_{5m}$ and $\by=a\bz_{3m}-b\bz_{m}-c\bz_{5m}$, or 
\end{enumerate}
\item $n+1=4m$, $\bx=a\bz_{2m}+b\bz_{m}+c\bz_{3m}$ and $\by=a\bz_{2m}-b\bz_{m}-c\bz_{3m}$.
\end{enumerate}
In all cases above, $a,b,c\in\R\backslash\{0\}$ are such that $a^2+b^2+c^2=\|\bx\|^2$. Moreover, 
the minimum PST times in (1a), (1b) and (2) are $\frac{\pi}{2}$, $\frac{\pi}{\sqrt{3}}$ and $\frac{\pi}{\sqrt{2}}$, respectively.
\end{theorem}

\begin{proof}
Let $|\sigma_{\bx}(A)|\geq 3$ and $\sigma_{\bx}(A)$ be closed under algebraic conjugates. As $|\mu_j|<2$, Corollary \ref{Cor:specgap} yields $|\sigma_{\bx}(A)|=3$. By Theorem \ref{Thm:rc}, we get two cases. If $\sigma_{\bx}(A)\subseteq\mathbb{Z}$, then using the same argument as Case 1 of the proof of Theorem \ref{Thm:pstcycles} to show (1a). If $\sigma_{\bx}(A)=\{\frac{a}{2},\frac{1}{2}(a\pm b\sqrt{\Delta})\}$, then one may use the argument in Case 2 to obtain (1b) and (2).
\end{proof}

\begin{example}
\label{Ex:pnex}
Let $n+1=4m$ and consider $\bx=\sum_{j=0}^{\lfloor\frac{m-1}{2}\rfloor}\be_{8j+1}-\sum_{j=1}^{\lfloor\frac{m}{2}\rfloor}\be_{8j-1}$ in $P_{n}$. Observe that we can write $\bx=\frac{\sqrt{n+1}}{2\sqrt{2}}(\bz_{2m}+\frac{1}{\sqrt{2}}(\bz_{m}+\bz_{3m}))$. Invoking Theorem \ref{Thm:pstpaths}(2) with $a=\sqrt{2}b=\sqrt{2}c=\frac{\sqrt{n+1}}{2\sqrt{2}}$, we get PST between $\bx$ and $\by=\frac{\sqrt{n+1}}{2\sqrt{2}}(-\bz_{2m}+\frac{1}{\sqrt{2}}(\bz_{m}+\bz_{3m}))=\sum_{j=0}^{\lfloor\frac{m-1}{2}\rfloor}\be_{8j+3}-\sum_{j=1}^{\lfloor\frac{m}{2}\rfloor}\be_{8j-3}$ at time $\frac{\pi}{\sqrt{2}}$. In particular, we have:
\begin{enumerate}
\item If $m=2$, then we obtain pair PST between $\bx=\be_1-\be_7$ and $\by=\be_3-\be_5$ in $P_7$.
\item If $m=3$, then we obtain PST between $\bx=\be_1-\be_7+\be_9$ and $\by=\be_3-\be_5+\be_{11}$ in $P_{11}$.
\end{enumerate}
\end{example}

\begin{corollary}
\label{Cor:pairpstpn}
Pair perfect state transfer occurs in $P_n$ relative to $A$ if and only if $n\in\{3,5,7\}$.
\end{corollary}

\begin{proof}
Suppose pair PST occurs in $P_n$ between $\be_u-\be_w$ and $\be_v-\be_x$. Since $\{u,w\}\neq \{v,x\}$, we have $n\geq 3$. From the proof of Theorem \ref{Thm:pstpaths}, we have $|\sigma_{\be_u-\be_w}(A)|\leq 3$. We proceed with two cases.

\noindent \textbf{Case 1.} Let $|\sigma_{\be_u-\be_w}(A)|=2$. That is, $\be_u-\be_w=a\bz_j+b\bz_k$ and $\be_v-\be_x=a\bz_j-b\bz_k$ with $j\neq k$. Adding these two equations yields $\be_u+\be_v-\be_w-\be_x=2a\bz_j$. Note that the $\ell$th entry of $\bz_j$ is nonzero if and only if $\sin\left(\frac{j\ell\pi}{n+1}\right)\neq 0$, or equivalently, $n+1$ does not divide $j\ell$. We now determine the conditions such that $\bz_j$ has at most four nonzero entries. If $n+1\in\{7,9,11\}$ or $n+1\geq 13$, then the Euler totient function yields at least five integers in $\{1,\ldots,n\}$ that are relatively prime to $n+1$, and so there are at least five values of $\ell$ such that $n+1$ does not divide $j\ell$. In this case, $\bz_j$ has at least five nonzero entries. Now, if $n+1\in\{8,10,12\}$, then for all $j\in\{1,\ldots,7\}$, one checks that there are at least five values of $\ell$ such that $n+1$ does not divide $j\ell$. In this case, we again get that $\bz_j$ has at least five nonzero entries. For the remaining cases $n\in\{3,4,5\}$, it is easy to check that there is PST between $\be_1-\be_2$ and $\be_3-\be_{2}$ in $P_3$ at $\frac{\pi}{\sqrt{2}}$, and between $\be_1-\be_5$ and $\be_2-\be_{4}$ in $P_5$ at $\frac{\pi}{2}$. The latter was also observed in \cite{pal2024quantum}.

\noindent \textbf{Case 2.} Let $|\sigma_{\be_u-\be_w}(A)|=3$. Theorem \ref{Thm:pstpaths} allows to us write $\be_u-\be_w=a\bz_{\frac{n+1}{2}}+b\bz_j+c\bz_{k}$ and $\be_v-\be_x=a\bz_{\frac{n+1}{2}}-b\bz_j-c\bz_{k}$ where $j,k\neq \frac{n+1}{2}$. Adding these two equations yields $\be_u+\be_v-\be_w-\be_x=2a\sqrt{\frac{2}{n+1}}[1,0,-1,0,1,\ldots]^\top$. This holds if and only if $u,v\in \{1,5\}$, $w,x\in \{3,7\}$, $n=7$ and $a=\frac{1}{2}$. Now, $\be_1-\be_3$ is not periodic in $P_7$, while pair PST happens between $\be_1-\be_7$ and $\be_3-\be_{5}$ at time $\frac{\pi}{\sqrt{2}}$ by Example \ref{Ex:pnex}(1). 

Combining the two cases above proves the forward direction. The converse is straightforward.
\end{proof}


An analogous argument yields $P_3$ as the only path that admits PST, between $\be_1+\be_2$ and $\be_3+\be_{2}$ in $P_3$ at $\frac{\pi}{\sqrt{2}}$.

\begin{corollary}
\label{Cor:pluspstpn}
Plus perfect state transfer occurs in $P_n$ relative to $A$ if and only if $n=3$. 
\end{corollary}

Despite the rarity of vertex, pair and plus PST in $P_n$ relative to $A$, Theorem \ref{Thm:pstpaths} guarantees that there are infinite families of paths that admit PST between real pure states.

We now turn to the Laplacian case. The Laplacian eigenvalues and eigenvectors of $P_n$ are known, see \cite[Section 1.4.4]{brouwer2011spectra}. We provide normalized eigenvectors for $L(P_n)$ below.

\begin{lemma}
\label{Lem:pnevalL}
The Laplacian eigenvector of $P_n$ corresponding to the eigenvalue $\theta_j=2\left(1-\cos\left(\frac{j\pi}{n}\right)\right)$ is
\begin{center}
$\bw_j=\sqrt{\frac{2}{n}}\left[\cos\left(\frac{j\pi}{2n}\right), \cos\left(\frac{3j\pi}{2n}\right),\cos\left(\frac{5j\pi}{2n}\right),\ldots,\cos\left(\frac{(2n-3)j\pi}{2n}\right), \cos\left(\frac{(2n-1)j\pi}{2n}\right)\right]^\top$
\end{center}
for $j\in\{0,1,\ldots n-1\}$ and $\bw_0=\frac{1}{\sqrt{n}}\ones$. Moreover, $\{\bw_0,\ldots,\bw_{n-1}\}$ is an orthonormal basis for $\R^n$.
\end{lemma}

The same argument in the proof of Theorem \ref{Thm:pstpaths} yields an analogous result for the Laplacian case.

\begin{theorem}
Let $n\geq 3$. Suppose $\bx,\by\in\R^n$ such that $|\sigma_{\bx}(L)|\geq 3$ and $\sigma_{\bx}(L)$ is closed under algebraic conjugates. Then $P_n$ admits Laplacian perfect state transfer between $\bx$ and $\by$ if and only if either:
\begin{enumerate}
\item $n=3m$, $\bx=a\bw_{2m}+b\bw_{\frac{3m}{2}}+c\bw_{m}+d\bw_{0}$, $\by=-a\bw_{2m}+b\bw_{\frac{3m}{2}}-c\bw_{m}+d\bw_{0}$, and $m$ is even if $b\neq 0$.
\item $n=6m$, $\bx=a\bw_{\frac{3m}{2}}+b\bw_{\frac{m}{2}}+c\bw_{\frac{5m}{2}}$ and $\by=a\bw_{\frac{3m}{2}}-b\bw_{\frac{m}{2}}-c\bw_{\frac{5m}{2}}$. 
\item $n=4m$, $\bx=a\bw_{2m}+b\bw_{m}+c\bw_{3m}$ and $\by=a\bw_{2m}-b\bw_{m}-c\bw_{3m}$.
\end{enumerate}
In all cases, $a,b,c,d\in\R$ are such that $a^2+b^2+c^2+d^2=\|\bx\|^2$ in (1) and $a^2+b^2+c^2=\|\bx\|^2$ otherwise. Moreover, the minimum PST times in (2) and (3) are $\frac{\pi}{\sqrt{3}}$ and $\frac{\pi}{\sqrt{2}}$, respectively, and $\pi$ otherwise.
\end{theorem}

Adapting the proof of Corollary \ref{Cor:pairpstpn} for the Laplacian case yields the following result.

\begin{corollary}
Laplacian pair perfect state transfer occurs in $P_n$ relative to $L$ if and only if $n\in\{3,4\}$. Meanwhile, Laplacian plus perfect state transfer occurs in $P_n$ if and only if $n=4$.
\end{corollary}

\begin{remark}
\label{Rem:pnL}
We make the following observations about $P_n$ for $n\in\{3,4,5\}$ relative to $L$.
 \begin{enumerate}
\item In $P_3$, $\{\be_1-\be_2,\be_3-\be_2\}$ has PST at $\frac{\pi}{2}$. Moreover, $\be_1+\be_2$ and $\be_3+\be_2$ are strongly cospectral and periodic with $\sigma_{\bx,\by}^+(L)=\{0,3\}$ and $\sigma_{\bx,\by}^-(L)=\{1\}$, but they do not admit PST as Corollary \ref{Cor:PSTcharcor}(3) does not hold.
\item In $P_4$, $\{\be_1+\be_4,\be_2+\be_3\}$ has PST at $\frac{\pi}{2}$, $\{\be_1-\be_2,\be_3-\be_4\}$ has PST at $\frac{\pi}{\sqrt{2}}$, and $\{\be_2-\be_3,\frac{1}{\sqrt{2}}(\be_1-\be_2+\be_3-\be_4)\}$ and $\{\be_1-\be_4,\frac{1}{\sqrt{2}}(\be_1+\be_2-\be_3-\be_4)\}$ have PST at $\frac{\pi}{2\sqrt{2}}$.
\item In $P_5$, $\{\be_1-\be_5,\frac{1}{\sqrt{5}}(\be_1+2\be_2-2\be_4-\be_5)\}$ and $\{\be_2-\be_4,\frac{1}{\sqrt{5}}(2\be_1-\be_2+\be_4-2\be_5)\}$ have PST, both at $\frac{\pi}{\sqrt{5}}$.
\end{enumerate}
\end{remark}

In \cite{Chen2020PairST}, Laplacian PST between pair states that form edges, also known as edge PST, was characterized for $P_n$. It turns out, $P_n$ admits Laplacian pair PST if and only if it admits Laplacian edge PST.

We end this section with a remark about the PST time between real pure states in paths.

\begin{remark}
\label{pstimepath}
Let $\tau_n$ be the least minimum PST time in $P_n$. Relative to $A$, we have $\tau_n=\frac{\pi}{4\cos(\frac{\pi}{n+1})}$, attained by $\bx=a\bz_1+b\bz_{n}$ and $\by=a\bz_1-b\bz_{n}$.
Relative to $L$, we have $\tau_n=\frac{\pi}{2\left(1-\cos\left(\frac{\left(n-1\right)\pi}{n}\right)\right)}$ attained by $\bx=a\bw_0+b\bw_{n-1}$ and $\by=a\bw_0-b\bw_{n-1}$.
In both cases, $\tau_n\rightarrow \frac{\pi}{4}$ as $n\rightarrow\infty$.
\end{remark}

\section{Cartesian product}\label{Sec:cartprod}

In this section, we use the Cartesian product to construct larger graphs that admit PST between real pure states. Let $G$ and $H$ be weighted graphs on $m$ and $n$ vertices, respectively. The Cartesian product of $G$ and $H$, denoted $G\square H$, is the graph with vertex set $V(G)\times V(H)$ such that $$M(G\square H)=M(G)\otimes I_n+I_m\otimes M(H),$$
where $M\in\{A,L\}$. Hence, $U_{M(G\square H)}(t)=U_{M(G)}(t)\otimes U_{M(H)}(t)$, from which we obtain $$U_{M(G\square H)}(t)(\bx\otimes \by)=U_{M(G)}(t)\bx\otimes U_{M(H)}(t)\by.$$
From the above equation, we get the following result which holds for $M\in\{A,L\}$.

\begin{theorem}
\label{Thm:cartprod}
Let $\bx_1,\by_1\in \R^m$ and $\bx_2,\by_2\in \R^n$ such that $\bx_1\neq \pm \by_1$. Then perfect state transfer occurs between $\bx_1\otimes \bx_2$ and $\by_1\otimes \by_2$ in $G\square H$ at time $\tau$ if and only if either
\begin{enumerate}
\item $\bx_2\neq \pm \by_2$, perfect state transfer occurs between $\bx_2$ and $\by_2$ in $H$, and perfect state transfer occurs between $\bx_1$ and $\by_1$ in $G$ both at time $\tau$; or
\item $\bx_2=\pm \by_2$ is periodic in $H$ and perfect state transfer occurs between $\bx_1$ and $\by_1$ in $G$ both at time $\tau$.
\end{enumerate}
\end{theorem}

\begin{corollary}
\label{Cor:cartconst}
Let $r,s\in\R\backslash\{0\}$ with $r=\pm s$.
If $H$ admits perfect state transfer between $\be_a+s\be_b$ and $\be_c+r\be_d$ at time $\tau$, where $\{a,b\}\neq \{c,d\}$ whenever $r=s$, then the following hold.
\begin{enumerate}
\item If $G$ admits perfect state transfer between vertex states $\be_u$ and $\be_v$ at time $\tau$, then $G\square H$ admits perfect state transfer between the pure states  $\be_u\otimes(\be_a+s\be_b)$ and $\be_v\otimes(\be_c+r\be_d)$ at time $\tau$.
\item If $G$ is periodic at a vertex states $\be_u$ at time $\tau$, then $G\square H$ admits perfect state transfer between the pure states  $\be_u\otimes(\be_a+s\be_b)$ and $\be_u\otimes(\be_c+r\be_d)$ at time $\tau$.
\end{enumerate}
\end{corollary}

Since $\be_u\otimes(\be_a+s\be_b)$ is an $s$-pair state in $G\square H$, the above corollary with $r=s$ can be used to construct infinite families of graphs with $s$-pair PST.

\begin{example}
Consider the hypercube $Q_d$ of dimension $d\ge 1$, which is known to admit PST between any pair of antipodal vertices $\be_u$ and $\be_{v}$ at time $\frac{\pi}{2}$. The following hold for all $d\geq 1$.
\begin{itemize}
\item By Corollary \ref{Cor:cartconst}(1) and Example \ref{Ex:c12}, $Q_d\square C_8$ admits PST between $\{\be_u\otimes(\be_0+\be_4),\be_v\otimes(\be_2+\be_6)\}$, and $Q_d\square C_{12}$ admits PST between $\{\be_u\otimes(\be_0+\be_4+\be_{8}),\be_v\otimes(\be_2+\be_6+\be_{10})\}$ at $\frac{\pi}{2}$.
\item By Corollary \ref{Cor:cartconst} and \cite[Theorem 6.5(iv-vi)]{kim2024generalization}, $Q_d\square C_6$ admits PST between $\{\be_u\otimes(\be_0-\be_2),\be_v\otimes(\be_3-\be_5)\}$ at $\frac{\pi}{2}$, and between $\{\be_u\otimes(\be_0+2\be_2),\be_u\otimes(\be_0+2\be_4)\}$ and $\{\be_u\otimes(\be_0+\frac{1}{2}\be_2),\be_u\otimes(\be_4+\frac{1}{2}\be_2)\}$ at $\pi$.
\item By Corollary \ref{Cor:cartconst}(1) and Remark \ref{Rem:pnL}(1,2), we get PST in $Q_d\square P_3$ and $Q_d\square P_4$  relative to $L$ between the pairs $\{\be_u\otimes(\be_1-\be_2),\be_v\otimes(\be_3-\be_2)\}$ and $\{\be_u\otimes(\be_1+\be_4),\be_v\otimes(\be_2+\be_3)\}$ at $\frac{\pi}{2}$, respectively.
\end{itemize}
\end{example}

\begin{example}
Let $M=A$ and consider $P_3^{\square n}$, the Cartesian product of $n\geq 1$ copies of $P_3$. This graph admits PST at time $\frac{\pi}{\sqrt{2}}$ between any pair of vertices $\be_u$ and $\be_{v}$ at distance $2n$. By Corollary \ref{Cor:cartconst}(1) and Example \ref{Ex:pnex}(1), $(P_3^{\square n})\square P_7$ admits PST between $\{\be_u\otimes(\be_1-\be_7),\be_v\otimes(\be_3-\be_5)\}$ at $\frac{\pi}{\sqrt{2}}$ for all $n\geq 1$. Moreover, by Corollary \ref{Cor:cartconst}(1) with \cite[Theorem 6.5(vii)]{kim2024generalization}, $(P_3^{\square n})\square C_8$ admits PST between $\{\be_u\otimes(\be_0-\be_2),\be_v\otimes(\be_4-\be_6)\}$ at $\frac{\pi}{\sqrt{2}}$ for all $n\geq 1$.
\end{example}

\section{Joins}\label{Sec:joins}

Let $G$ and $H$ be weighted graphs on $m$ and $n$ vertices, respectively. The \textit{join} of $G$ and $H$, denoted $G\vee H$, is obtained from taking a copy of $G$ and a copy of $H$ and adding all edges between $G$ and $H$ with weight one. Throughout, we assume that $G$ and $H$ are $k$- and $\ell$-regular graphs, respectively whenever we deal with $M=A$.

Let $\lambda$ and $\mu$ be nonzero eigenvalues of $L(G)$ and $L(H)$, respectively. From \cite[Equation 33]{Alvir2016}, the transition matrix of $G\vee H$ relative to $L$ is given by
\begin{equation}
\label{Eq:transmatrixL}
\begin{split}
U_{L}(t)&=\frac{1}{m+n}J+\frac{e^{it(m+n)}}{mn(m+n)}\left[ \begin{array}{ccccc} n^2J&-mnJ \\ -mnJ&m^2J\end{array} \right]+\sum_{\lambda>0}e^{it(\lambda+n)}\left[\begin{array}{cc} E_{\lambda}&\zero \\ \zero&\zero\end{array} \right]+\sum_{\mu>0}e^{it(\mu+m)}\left[\begin{array}{cc} \zero&\zero \\ \zero&F_{\mu}\end{array} \right], 
\end{split}
\end{equation}
whenever $G$ and $H$ are connected. If $G$ (respectively, $H$) is disconnected, then we add the term $e^{itn}\left[\begin{array}{cc} E_{0}-\frac{1}{m}J&\zero \\ \zero&\zero\end{array}\right]$ (respectively, $e^{itm}\left[\begin{array}{cc} \zero&\zero \\ \zero&F_{0}-\frac{1}{n}J\end{array} \right]$) in the third (respectively, fourth) summand in (\ref{Eq:transmatrixL}).

Suppose further that $G$ and $H$ are $k$- and $\ell$-regular graphs, respectively. Let $\lambda<k$ and $\mu<\ell$ be eigenvalues of $A(G)$ and $A(H)$ respectively. Let $\lambda^{\pm }=\frac{1}{2}(k+\ell\pm \sqrt{\Delta})$, where $\Delta=(k-\ell)^2+4mn$. From \cite[Equation 12.2.1]{Coutinho2021}, the transition matrix of $G\vee H$ relative to $A$ is given by 
\begin{equation}
\label{Eq:transmatrixA}
\begin{split}
U_{A}(t)&=\frac{e^{it\lambda^+}}{m\sqrt{\Delta}(k-\lambda^-)}\bu\bu^\top+\frac{e^{it\lambda^-}}{m\sqrt{\Delta}(\lambda^+-k)}\bv\bv^\top+\sum_{\lambda<k}e^{it\lambda}\left[\begin{array}{cc} E_{\lambda}&\zero \\ \zero&\zero\end{array} \right]+\sum_{\mu<\ell}e^{it\mu}\left[\begin{array}{cc} \zero&\zero \\ \zero&F_{\mu}\end{array} \right],      
\end{split}
\end{equation}
whenever $G$ and $H$ are connected, where $\bu=\left[ \begin{array}{ccccc} (k-\lambda^{-})\ones_{m} \\ m\ones_{n}\end{array} \right]$ and $\bv=\left[ \begin{array}{ccccc} (k-\lambda^{+})\ones_{m} \\ m\ones_{n}\end{array} \right]$. If $G$ (respectively, $H$) is disconnected, then we include the term $e^{itk}\left[\begin{array}{cc} E_{0}-\frac{1}{m}J&\zero \\ \zero&\zero\end{array} \right]$ (respectively, $e^{it\ell}\left[\begin{array}{cc} \zero&\zero \\ \zero&F_{0}-\frac{1}{n}J\end{array} \right]$) in the third (respectively, fourth) summand in equation (\ref{Eq:transmatrixA}). For more about quantum walks on join graphs, see \cite{kirkland2023quantum}.

The join operation can be used to construct larger graphs that admit PST between real pure states.

\begin{theorem}
Let $\bx_1,\by_1\in\R^{m}$ and $\bx_2,\by_2\in\R^{n}$ be unit vectors such that $\ones^\top\bx_1=\ones^\top\bx_2=0$. If $G$ and $H$ are connected, then the following hold relative to $M\in\{A,L\}$.
\begin{enumerate}
\item $G$ admits perfect state transfer between $\bx_1$ and $\by_1$ if and only if $G\vee H$ admits perfect state transfer between $\left[ \begin{array}{cc} \bx_1 \\ \zero\end{array} \right]$ and $\left[ \begin{array}{cc} \by_1 \\ \zero\end{array} \right]$ at the same time.
\item Suppose $G$ admits perfect state transfer between $\bx_1$ and $\by_1$ and $H$ admits perfect state transfer between $\bx_2$ and $\by_2$ both at time $\tau$. If $\tau(\lambda-\theta+\delta(n-m))\equiv 0$ (mod $2\pi$) for some $\lambda\in\sigma_{\bx_1,\by_1}^+(M)$, $\theta\in\sigma_{\bx_2,\by_2}^+(M)$ and $\delta\in\{0,1\}$, then $G\vee H$ admits perfect state transfer between $\left[ \begin{array}{cc} \bx_1 \\ \bx_2\end{array} \right]$ and $\left[ \begin{array}{cc} \by_1 \\ \by_2\end{array} \right]$ at time $\tau$ relative to $M=L$ whenever $\delta=1$ and relative to $M=A$ whenever $\delta=0$.
\end{enumerate}
\end{theorem}

\begin{proof}
Let $U_G(t)$ denote the transition matrix of $G$ relative to $L$. As $\ones^\top\bx_1=0$, equation (\ref{Eq:transmatrixL}) yields
\begin{center}
$U_{L}(t)\left[\begin{array}{cc} \bx_1 \\ \zero\end{array} \right]=\left[\begin{array}{cc} e^{it n}U_G(t)\bx_1 \\ \zero\end{array} \right]\quad$ and $\quad U_{A}(t)\left[\begin{array}{cc} \bx_1 \\ \zero\end{array} \right]=\left[\begin{array}{cc} U_G(t)\bx_1 \\ \zero\end{array} \right]$.
\end{center}
From these equations, (1) is straightforward. We now prove (2). Since $\ones^\top\bx_1=\ones^\top\bx_2=0$, the same argument yields \begin{equation}
\label{Eq:Ljoin}
U_{L}(t)\left[\begin{array}{cc} \bx_1 \\ \bx_2\end{array} \right]=\left[\begin{array}{cc} e^{it n}U_G(t)\bx_1 \\ \zero\end{array} \right]+\left[\begin{array}{cc} \zero \\ e^{it m}U_H(t)\bx_2\end{array} \right]=\left[\begin{array}{cc} e^{it n}U_G(t)\bx_1 \\ e^{it m}U_H(t)\bx_2\end{array} \right].
\end{equation}
As PST occurs between $\bx_j$ and $\by_j$ for $j\in\{1,2\}$ at $\tau$, we have $U_G(\tau)\bx_1=\gamma_1\by_1$ and $U_H(\tau)\bx_2=\gamma_2\by_2$, where $\gamma_1=e^{i\tau\lambda}$ and $\gamma_2=e^{i\tau\theta}$ for all $\lambda\in\sigma_{\bx_1,\by_1}^+(L)$ and $\theta\in\sigma_{\bx_2,\by_2}^+(L)$. Thus, if $\tau(\lambda-\theta+n-m)\equiv 0$ (mod $2\pi$) for some $\lambda\in\sigma_{\bx_1,\by_1}^+(M)$ and $\theta\in\sigma_{\bx_2,\by_2}^+(M)$, then $\tau(n+\lambda)\equiv \tau(m+\theta)$ (mod $2\pi$), and so $e^{i\tau (n+\lambda)}=e^{i\tau (m+\theta)}$. Thus, $U_{L}(\tau)\left[\begin{array}{cc} \bx_1 \\ \bx_2\end{array} \right]=\left[\begin{array}{cc} e^{i\tau n}U_G(\tau)\bx_1 \\ e^{i\tau m}U_H(\tau)\bx_2\end{array} \right]=\left[\begin{array}{cc} e^{i\tau (n+\lambda)}\by_1 \\ e^{i\tau (m+\theta)}\by_2\end{array} \right]=e^{i\tau(n+\lambda)}\left[\begin{array}{cc}\by_1 \\ \by_2\end{array} \right]$ by equation (\ref{Eq:Ljoin}). This proves that PST occurs between $\left[ \begin{array}{cc} \bx_1 \\ \bx_2\end{array} \right]$ and $\left[ \begin{array}{cc} \by_1 \\ \by_2\end{array} \right]$ in $G\vee H$. The same argument applies to $M=A$ except that $e^{itn}$ and $e^{itm}$ in (\ref{Eq:Ljoin}) are both absent.
\end{proof}

\section{Complete bipartite graphs}\label{Sec:cbp}

In this section, we characterize PST between real pure states in the complete bipartite graph $K_{m,n}$. We only focus on the case when $|\sigma_{\bx}(M)|\geq 3$, starting with the adjacency case.

\begin{theorem}
\label{Thm:kmnpst}
Let $\bx=\begin{bmatrix}
\bx_1 \\ \bx_2
\end{bmatrix},\by\in\R^{m+n}$, where $\bx_1\in\R^m$ and $|\sigma_{\bx}(A)|\geq 3$. Let $\bz$ and $\bv^{\pm}=\begin{bmatrix}\sqrt{n}\ones_m \\ \pm \sqrt{m}\ones_n\end{bmatrix}$ be eigenvectors for $A(K_{m,n})$ associated with 0 and $\pm \sqrt{mn}$ respectively. Then $\bx$ and $\by$ admit adjacency perfect state transfer in $K_{m,n}$ if and only if $\bx\in\operatorname{span}\{\bv^+,\bv^-,\bz\}$, $\bx\notin\operatorname{span}\mathcal{U}$ for any two-subset $\mathcal{U}$ of $\{\bv^+,\bv^-,\bz\}$ and $\by=\bx-2\begin{bmatrix}
\frac{1}{m}(\ones^{\top}\bx_1)\ones \\ \frac{1}{n}(\ones^{\top}\bx_2)\ones
\end{bmatrix}$. Moreover, the minimum PST time is $\frac{\pi}{\sqrt{mn}}$.

\end{theorem}

\begin{proof}
Since $K_{m,n}=O_m\vee O_n$,
(\ref{Eq:transmatrixA}) yields the following spectral decomposition for $A(K_{m,n})$:
\begin{equation}
\label{Eq:completebip}
U_A(t)=\begin{bmatrix}
I_m-\frac{1}{m}J&0 \\ 0&I_n-\frac{1}{n}J
\end{bmatrix}+\frac{e^{it\sqrt{mn}}}{2mn}\begin{bmatrix}
nJ&\sqrt{mn}J \\ \sqrt{mn}J&mJ
\end{bmatrix}+\frac{e^{-it\sqrt{mn}}}{2mn}\begin{bmatrix}
nJ&-\sqrt{mn}J \\ -\sqrt{mn}J&mJ
\end{bmatrix}.
\end{equation}
Taking $\by=U_A(\frac{\pi}{\sqrt{mn}})\bx$ yields the desired conclusion.
\end{proof}

\begin{corollary}
\label{Cor:pairpstkn}
Pair perfect state transfer occurs in $K_{m,n}$ relative to $A$ if and only if either
\begin{enumerate}
\item $(m,n)\in\{(1,2),(2,1)\}$, between $\be_u-\be_w$ and $\be_u-\be_x$, where $u$ is a degree two vertex.
\item $m=n=2$, between $\be_u-\be_w$ and $\be_v-\be_x$, where $\{u,w\}$ and $\{v,x\}$ are non-incident edges.  
\end{enumerate} 
\end{corollary}

\begin{proof}
Note that $\be_u-\be_w$ is fixed whenever $u$ and $w$ are non-adjacent. Now, suppose $u$ and $w$ are adjacent. Then equation (\ref{Eq:completebip}) yields $U_A(t)(\be_u-\be_w)=\begin{bmatrix}\be_u-\frac{2}{m} \ones\\ -(\be_w-\frac{2}{n}\ones)\end{bmatrix}$, and so (1) and (2) are immediate.
\end{proof}

\begin{corollary}
\label{Cor:pluspstkn}
Plus perfect state transfer occurs in $K_{m,n}$ relative to $A$ if and only if either (i) one of the two conditions in Corollary \ref{Cor:pairpstkn} hold with the pair states turned into plus states, or (ii) $m=4$ or $n=4$, between $\be_u+\be_w$ and $\be_v+\be_x$, where $\{u,w,v,x\}$ is a bipartition of size four.
\end{corollary}

\begin{proof}
If $u$ and $w$ are adjacent, then the same argument in Corollary \ref{Cor:pairpstkn} proves statement (i). Otherwise, equation (\ref{Eq:completebip}) yields $U_A(t)(\be_u+\be_w)=\begin{bmatrix}\be_u+\be_w-\frac{4}{m}\ones \\ \zero\end{bmatrix}$, and so statement (ii) is immediate.
\end{proof}

The minimum PST time in Corollaries \ref{Cor:pairpstkn} and \ref{Cor:pluspstkn} is $\frac{\pi}{\sqrt{mn}}$.

For the Laplacian case, equation (\ref{Eq:transmatrixL}) gives us the transition matrix for $K_{m,n}$:
$$U_{L}(t)=\frac{1}{m+n}J+\frac{e^{it(m+n)}}{mn(m+n)}\left[\begin{array}{cc} n^2J&-mnJ \\ -mnJ&m^2J\end{array} \right]+e^{itn}\left[\begin{array}{cc} I_m-\frac{1}{m}J&\zero \\ \zero&\zero\end{array} \right]+e^{itm}\left[\begin{array}{cc} \zero&\zero \\ \zero&I_n-\frac{1}{n}J\end{array} \right].$$
Using the same argument in the proof of Theorem \ref{Thm:kmnpst} yields an analogous result for the Laplacian case.

\begin{theorem}
\label{Thm:kmnpstlap}
Let $\bx=\begin{bmatrix}
\bx_1 \\ \bx_2
\end{bmatrix},\by\in\R^{m+n}$, where $\bx_1\in\R^m$ and where $|\sigma_{\bx}(L)|\geq 3$. Let $\bu$, $\bv$ and $\bw$ be eigenvectors for $L(K_{m,n})$ associated with $m$, $n$ and $m+n$ respectively. Then $\bx$ and $\by$ admit Laplacian perfect state transfer in $K_{m,n}$ if and only if $\bx\in\operatorname{span}\{\ones,\bu,\bv,\bw\}$, $\bx\notin\operatorname{span}\mathcal{U}$ for any two-subset $\mathcal{U}$ of $\{\ones,\bu,\bv,\bw\}$ and either
\begin{enumerate}
\item $\nu_2(m)=\nu_2(n)$ and $\by=\begin{bmatrix}
-\bx_1+\frac{2}{m} (\ones^\top\bx_1)\ones \\ -\bx_2+\frac{2}{n} (\ones^\top\bx_2)\ones\end{bmatrix}$
\item $\nu_2(m)>\nu_2(n)$ and $\by=\begin{bmatrix}-\bx_1+\frac{2}{m+n} ((\ones^\top\bx_2)+(\ones^\top\bx_1))\ones\\
\bx_2+\frac{2}{m+n}((\ones^\top\bx_1)-\frac{m}{n} (\ones^\top\bx_2))\ones\end{bmatrix}$.
\item $\nu_2(m)<\nu_2(n)$ and 
$\by=\begin{bmatrix}
\bx_1+\frac{2}{m+n}((\ones^\top\bx_2)-\frac{n}{m}(\ones^\top\bx_1))\ones \\ -\bx_2+\frac{2}{m+n}((\ones^\top\bx_1)+(\ones^\top\bx_2))\ones\end{bmatrix}$.
\end{enumerate}
The minimum PST time in all cases above is $\frac{\pi}{\operatorname{gcd}(m,n)}$.
\end{theorem}

We finish this section by characterizing Laplacian pair and plus PST in complete bipartite graphs.

\begin{corollary}
\label{Cor:pairkmn}
Laplacian pair perfect state transfer occurs in $K_{m,n}$ if and only if either Corollary \ref{Cor:pairpstkn}(2) holds or $(m,n)\in\{(2,4k),(4k,2)\}$ for any integer $k\geq 1$. In particular, perfect state transfer occurs between $\be_u-\be_w$ and $\be_v-\be_w$ in $K_{2,4k}$, where $\{u,v\}$ is a partite set of size two and $w\in V(K_{m,n})\backslash\{u,v\}$.
\end{corollary}

\begin{proof}
If $u$ and $w$ are non-adjacent, then $\be_u-\be_w$ is fixed. Otherwise, $\be_u-\be_w$ has eigenvalue support $\{m,n,m+n\}$. 
Applying Theorem \ref{Thm:kmnpstlap}(1) with $\bx_1=\be_u$ and $\bx_2=-\be_w$ yields the desired conclusion.
\end{proof}

In \cite{chen2019edge}, it was shown that $K_{2,4k}$ admits Laplacian pair PST. Thus, $C_4$ and $K_{2,4k}$ are the only complete bipartite graphs that admit pair PST by Corollary \ref{Cor:pairkmn}. For plus state transfer, we have the following:

\begin{corollary}
\label{Cor:pairkmn1}
Laplacian plus perfect state transfer occurs in $K_{m,n}$ if and only if either
\begin{enumerate}
\item $m=n=2$, between $\be_u+\be_w$ and $\be_v+\be_x$, where either (i) $\{u,w\}$ and $\{v,x\}$ are non-incident edges or (ii) $\{u,w\}$ and $\{v,x\}$ are the two partite sets of size two, or
\item $(m,n)\in\{(4,4k),(4k,4)\}$ for any odd $k$, between $\be_u+\be_w$ and $\be_v+\be_x$, where $\{u,w,v,x\}$ is a partite set of size four.
\end{enumerate}
\end{corollary}

\begin{proof}
First, suppose $u$ and $w$ are adjacent. If $\nu_2(m)=\nu_2(n)$, then Theorem \ref{Thm:kmnpstlap}(1) yields $\by=\begin{bmatrix}
\be_u-\frac{2}{m} \ones \\ \be_w-\frac{2}{n} \ones\end{bmatrix}$. This proves (1i). Now, if $\nu_2(m)>\nu_2(n)$, then Theorem \ref{Thm:kmnpstlap}(2) gives us $\by=\begin{bmatrix}
-\be_u+\frac{4}{m+n} \ones \\ \be_w+\frac{2}{m+n}(1-\frac{m}{n})\ones\end{bmatrix}$, which is not a plus state for any $m$ and $n$. Similarly for the case $\nu_2(m)<\nu_2(n)$. Now, suppose $u$ and $w$ are non-adjacent. If $\nu_2(m)=\nu_2(n)$, then Theorem \ref{Thm:kmnpstlap}(1) again yields $\by=\begin{bmatrix}
-\be_u-\be_w+\frac{4}{m} \ones \\ \zero\end{bmatrix}$. From this, (2) follows. If $\nu_2(m)>\nu_2(n)$, then Theorem \ref{Thm:kmnpstlap}(2) implies that $\by=\begin{bmatrix}
-\be_u-\be_w+\frac{4}{m+n} \ones \\ \frac{4}{m+n} \ones\end{bmatrix}$, which yields (1ii). If $\nu_2(m)<\nu_2(n)$, then one gets the same result by applying Theorem \ref{Thm:kmnpstlap}(3).
\end{proof}

\section{Minimizing PST time}\label{Sec:minpst}

The following result determines the vectors $\bx$ and graphs $G$ such that the minimum period of $\bx$ in $G$ is the least amongst all unweighted connected $n$-vertex graphs.

\begin{theorem}
\label{Thm:minperjoin}
Let $\bx\in \R^n$. The following hold.
\begin{enumerate}
\item Amongst all connected unweighted $n$-vertex graphs, $\bx$ attains the least minimum period in $G$ relative to $L$ if and only if $G=G_1\vee G_2$ with $|V(G_i)|=n_i$ for $i\in\{1,2\}$, $n=n_1+n_2,$ and $\bx\in\operatorname{span}\left\{\ones_n,\begin{bmatrix}
n_2\ones_{n_1} \\ -n_1\ones_{n_2}
\end{bmatrix}\right\}$.
\item There exists an integer $N>0$ such that for all connected unweighted $n$-vertex graphs with $n\geq N$, $\bx$ attains the least minimum period in $G$ relative to $A$ if and only if $G=O_a\vee K_{n-a}$  with $a=\lceil\frac{n}{3}\rceil$, and $\bx\in\operatorname{span}\left\{\begin{bmatrix}
-\lambda^-\ones_{a} \\ a\ones_{n-a}
\end{bmatrix},\begin{bmatrix}
-\lambda^+\ones_{a} \\ a\ones_{n-a}
\end{bmatrix}\right\}$, where $\lambda^{\pm}=\frac{1}{2}(n-a-1\pm \sqrt{(n-a-1)^2+4a(n-a)})$.
\end{enumerate}
Moreover, for 1 and 2, we have $\rho=\frac{2\pi}{n}$ and $\rho=\frac{2\pi}{\sqrt{(n-a-1)^2+4a(n-a)}}\approx\frac{\pi\sqrt{3}}{n}$, respectively.
\end{theorem}

\begin{proof}
Let $\lambda_1$ and $\lambda_2$ be the largest and smallest eigenvalues of $M$. We first prove 1. By assumption, $0$ is a simple eigenvalue of $L(G)$ with eigenvector $\ones_n$. Moreover, every eigenvalue $\lambda$ of $L(G)$ satisfies $\lambda\leq n$ with equality if and only if $G$ is a join graph. Thus, for any two eigenvalues $\lambda_1$ and $\lambda_2$ of $L(G)$, the Laplacian spread $\lambda_1-\lambda_2$ is maximum if and only if $\lambda_1=n$ and $\lambda_2=0$. 
Invoking Lemma \ref{Lem:minperiod}, the least minimum period is attained if and only if $G=G_1\vee G_2$ for some graphs $G_i$ on $n_i$ vertices, $i\in\{1,2\}$ and $\sigma_{\bx}(L)=\{0,n\}$, in which case $\rho=\frac{2\pi}{n}$ and $\begin{bmatrix}
n_2\ones_{n_1} \\ -n_1\ones_{n_2}
\end{bmatrix}$ is the eigenvector associated with $n$. To prove 2, we use a result due to Breen, Riasanovsky, Tait and Urschel \cite{BRTU} states that there is an $N>0$ such that if $n\ge N,$ the maximum adjacency spread $\lambda_1-\lambda_2$ over all connected $n$-vertex graphs is attained uniquely by the complete split graph $G=O_a\vee K_{n-a}.$ In this case, we have $\lambda_1=\lambda^+$, $\lambda_2=\lambda^-$ and $\lambda^{\pm}=\frac{1}{2}(n-a-1\pm\sqrt{(n-a-1)^2+4a(n-a)})$ so that $\lambda_{1}-\lambda_{2}=\sqrt{(n-a-1)^2+4a(n-a)} \approx \frac{2n}{\sqrt{3}}.$ The same argument used in the above case yields the desired conclusion.
\end{proof}

\begin{corollary}
\label{Cor:minpst}
Let $\bx,\by\in \R^n$.
\begin{enumerate}
\item Amongst all connected unweighted $n$-vertex graphs, $\bx$ and $\by$ attain the least minimum PST time relative to $L$ if and only if the conditions in Theorem \ref{Thm:minperjoin}(1) hold and
\begin{center}
$\by=\frac{1}{n}(\ones_{n}^\top \bx)\ones_{n}-\frac{1}{n_1n_2(n_1+n_2)}\left(\begin{bmatrix}
n_2\ones_{n_1} \\ -n_1\ones_{n_2}
\end{bmatrix}^\top\bx\right)\begin{bmatrix}
n_2\ones_{n_1} \\ -n_1\ones_{n_2}
\end{bmatrix}.$        
\end{center}
\item There exists an integer $N>0$ such that amongst all connected unweighted $n$-vertex graphs with $n\geq N$, $\bx$ and $\by$ attain the least minimum PST time relative to $A$ if and only if the conditions in Theorem \ref{Thm:minperjoin}(2) hold and $\by$ has the form below, where $D^{\pm}=\pm a(\lambda^{\pm })\sqrt{(n-a-1)^2+4a(n-a)}$.
\begin{center}
$\by=\frac{1}{D^-}\left(\begin{bmatrix}
-\lambda^-\ones_{a} \\ a\ones_{n-a}
\end{bmatrix}^\top\bx\right)\begin{bmatrix}
-\lambda^-\ones_{a} \\ a\ones_{n-a}
\end{bmatrix}-\frac{1}{D^+}\left(\begin{bmatrix}
-\lambda^+\ones_{a} \\ a\ones_{n-a}
\end{bmatrix}^\top\bx\right)\begin{bmatrix}
-\lambda^+\ones_{a} \\ a\ones_{n-a}
\end{bmatrix}$.
\end{center}
\end{enumerate} 
The minimum PST time in (1) is $\frac{\pi}{n}$, and (for all sufficiently large $n$) $\frac{\pi}{\sqrt{(n-a-1)^2+4a(n-a)}},$ otherwise.
\end{corollary}

\begin{proof}
This follows from Theorems \ref{Thm:PSTchar}(1) and \ref{Thm:minperjoin}, and the fact that $|\sigma_{\bx}(M)|=2$.
\end{proof}

We close this section with the following remark.

\begin{remark}
As $n\rightarrow\infty$, the least minimum PST times in Corollary \ref{Cor:minpst} both tend to 0, in contrast to the least minimum PST time for paths which tend to $\frac{\pi}{4}$ by Remark \ref{pstimepath}. Thus, the join graphs in Corollary \ref{Cor:minpst} are desirable if a smaller minimum PST time is preferred.
\end{remark}

\section{Sensitivity with respect to  readout time}\label{Sec:sens}
 
Suppose PST occurs between real unit vectors $\bx, \by \in \R^n$ relative to $M$ at time $\tau$. Define $f:\R^+\rightarrow [0,1]$ as $$f(t) = |\by^\top e^{itM}\bx|^2.$$
If $\bx$ and $\by$ are not unit vectors, then we may define the above function as $f(t) =\frac{1}{\|\bx\|^2} |\by^\top e^{itM}\bx|^2$. This adds a constant factor of $\frac{1}{\|\bx\|^2}$ to the above function, and so in order to simplify our calculations, we shall only deal with real unit vectors in this section.

Note that $f(t)$ is the analogue of the fidelity of state transfer from $\bx$ to $\by$ at time $t$. Following the proof of Theorem 2.2 in \cite{Ksens}, we find that for each $k \in \mathbb{N},$ 
\begin{equation}
\label{Eq:f2}
\frac{d^kf}{dt^k}\Big|_{\tau} = \begin{cases} (-1)^{\frac{k \mod 4}{2}} \displaystyle\sum_{j=0}^k(-1)^j {{k}\choose{j}} (\by^\top M^j \by)( \by^\top M^{k-j}\by) & {\rm{if}} \  k \ \rm{is \ even}\\
0 & {\rm{if}} \  k \ \rm{is \ odd}. \end{cases}  
\end{equation}
Moreover, by Lemma \ref{Prop:cosp}, we have $\bx^\top M^k \bx= \by^\top M^k \by$ for all integers $k\geq 0$.

There is practical interest in determining how large or small  $ \frac{d^2f}{dt^2}\Big|_{\tau}$ can be, since that will give intuition on what features govern the sensitivity of the fidelity with respect to the readout time. Taking $k=2$ in (\ref{Eq:f2}) yields 
\begin{equation}
\label{Eq:der2} 
\begin{split}
\frac{d^2f}{dt^2}\Big|_{\tau}=-2(\by^\top M^2 \by - (\by^\top M \by)^2).
\end{split}
\end{equation}
Let $\sigma_{\by}(M)=\{\lambda_1, \ldots, \lambda_\ell\}$ and write $\by=\sum_{j=1}^\ell c_j\bv_j$ where each $\bv_j$ is a unit vector associated with $\lambda_j$, each scalar $c_j\in\R\backslash\{0\}$ and $\sum_{j=1}^\ell c_j^2=1$. Then
\begin{equation*}
\frac{d^2f}{dt^2}\Big|_{\tau} = 2((\by^\top M \by)^2 -\by^\top M^2 \by  ) = 2\left(\left(\sum_{j=1}^\ell a_j \lambda_j\right)^2 - \sum_{j=1}^\ell a_j \lambda_j^2 \right),
\end{equation*}
where $a_j=c_j^2$ for each $j=1, \ldots, \ell.$ Observe that from the Cauchy-Scwarz inequality,
\begin{equation*}
\left(\sum_{j=1}^\ell a_j \lambda_j\right)^2 = \left(\sum_{j=1}^\ell (c_j \lambda_j) c_j\right)^2 \le \left(\sum_{j=1}^\ell c_j^2 \lambda_j^2\right) \left(\sum_{j=1}^\ell c_j^2\right) = \sum_{j=1}^\ell a_j \lambda_j^2.
\end{equation*}
We thus deduce that $\frac{d^2f}{dt^2}\Big|_{\tau} \le 0.$ This upper bound may be approached arbitrarily closely when 
$\bx$ is of the form $\epsilon\bu + \sqrt{1-\epsilon^2} \bv,$ where $\bu$ and $\bv$ 
are unit  eigenvectors corresponding to different eigenvalues, and $|\epsilon|>0$ is small. Furthermore, from (\ref{Eq:der2}), we obtain $\frac{d^2f}{dt^2}\Big|_{\tau}=0$ if and only if $\by$ is a unit eigenvector for $M$. However, this implies that $\by$ is fixed by Proposition \ref{Prop:fixedstatechar}. That is, $\by$ cannot exhibit strong cospectrality, a contradiction to the fact that $\by$ is involved in PST. Therefore, $\frac{d^2f}{dt^2}\Big|_{\tau}<0$.

Next we seek  the minimum value of  $(\sum_{j=1}^\ell a_j \lambda_j)^2-\sum_{j=1}^\ell a_j \lambda_j^2 $ subject to the constraint that $\sum_{j=1}^\ell a_j=1$ and $a_j \ge 0$ for all $j=1, \ldots, \ell.$ 
Suppose that we have distinct indices $j_1, j_2$ such that $a_{j_1}, a_{j_2}>0.$ For each $j=1, \ldots, \ell$, consider the coefficients $b_j$ given by $b_j=a_j$ whenever $j\ne j_1,j_2, b_{j_1}=a_{j_1}+\epsilon$, and $b_{j_2}=a_{j_2}-\epsilon,$ where $\epsilon$ is sufficiently small. It is straightforward to show that 
\begin{equation}
\label{der3}
\begin{split}
\left(\sum_{j=1}^\ell b_j \lambda_j\right)^2 -\sum_{j=1}^\ell b_j \lambda_j^2 &=
\left(\sum_{j=1}^\ell a_j \lambda_j\right)^2 - 
\sum_{j=1}^\ell a_j \lambda_j^2 +\epsilon(\lambda_{j_1}-\lambda_{j_2})\left(2\sum_{j=1}^\ell a_j \lambda_j -(\lambda_{j_1}+ \lambda_{j_2} )  \right) +\epsilon^2(\lambda_{j_1}-\lambda_{j_2})^2.        
\end{split}
\end{equation}
Observe that if there is a third index $j_3 \ne j_1, j_2$ such that $a_{j_3}>0,$ then because $\lambda_{j_2}\neq \lambda_{j_3}$, we get that one of $2\sum_{j=1}^\ell a_j \lambda_j -(\lambda_{j_1}+ \lambda_{j_2} )$ and $2\sum_{j=1}^\ell a_j \lambda_j -(\lambda_{j_1}+ \lambda_{j_3} )$ is nonzero. In particular, if $2\sum_{j=1}^\ell a_j \lambda_j -(\lambda_{j_1}+ \lambda_{j_2})\neq 0$, then by virtue of (\ref{der3}), we can choose a value of $\epsilon$ 
so that
\begin{equation*}
\left(\sum_{j=1}^\ell b_j \lambda_j\right)^2 -\sum_{j=1}^\ell b_j \lambda_j^2  < \left(\sum_{j=1}^\ell a_j \lambda_j\right)^2 - \sum_{j=1}^\ell a_j \lambda_j^2.
\end{equation*}
In this case, $(\sum_{j=1}^\ell a_j \lambda_j)^2 - \sum_{j=1}^\ell a_j \lambda_j^2$ cannot attain the minimum value. So in order to find the minimum value, it suffices to focus on expressions of the form $(a \lambda_{j_1} + (1-a) \lambda_{j_2})^2 -( a \lambda_{j_1}^2 + (1-a) \lambda_{j_2}^2 )$ where $a \in [0,1]$ 
and $\lambda_{j_1}, \lambda_{j_2} \in \{\lambda_1, \ldots, \lambda_\ell\}.$ Let $\lambda_{\max}, \lambda_{\min}$ denote the largest and smallest elements of $\{\lambda_1, \ldots, \lambda_\ell\},$ respectively.
Since
\begin{equation*}
(a \lambda_{j_1} + (1-a) \lambda_{j_2})^2 -( a \lambda_{j_1}^2 + (1-a) \lambda_{j_2}^2 ) = -a(1-a)(\lambda_{j_1}-\lambda_{j_2})^2,
\end{equation*} 
we find that the minimum value for $\frac{d^2f}{dt^2}\Big|_{\tau} $ is given by $-\frac{1}{2}(\lambda_{\max} - \lambda_{\min} )^2$, which is attained by $\bx = \frac{1}{\sqrt{2}}(\bu +\bv)$ and $\by =\frac{1}{\sqrt{2}}(\bu -\bv),$ where $\bu, \bv$ are unit eigenvectors corresponding to $\lambda_{\max} $ and $\lambda_{\min},$ respectively. 

We summarize the above discussion as follows. 

\begin{theorem}
\label{thm:sens}
Suppose perfect state transfer occurs between unit vectors $\bx$ and $\by$ at time $\tau$ relative to $M$. Let $\lambda_{\max} = \max \sigma_{\bx}(M), \lambda_{\min} = \min \sigma_{\bx}(M).$ Then $$0 > \
\frac{d^2f}{dt^2}\bigg|_{\tau}\ \ge 
 \ -\frac{1}{2}(\lambda_{\max} - \lambda_{\min} )^2.$$
\end{theorem}
 
\begin{example}
Consider the Petersen graph with adjacency matrix $A,$ which has eigenvalues $3, 1, -2$. Let $\bv_1, \bv_2, \bv_3$ be unit eigenvectors corresponding to $3, 1$ and $-2,$ respectively and form $\bx\in \R^{10}$ as $\bx=c_1\bv_1 + c_2 \bv_2+c_3\bv_3$ where $c_1, c_2, c_3\in \R\backslash\{0\}, c_1^2+c_2^2+c_3^2=1.$ Setting $\by=-c_1\bv_1 - c_2 \bv_2+c_3\bv_3$, we find that there is PST from $\bx\bx^\top$ to $\by\by^\top$ at time $\pi.$ According to Theorem \ref{thm:sens}, $\frac{d^2f}{dt^2}\Big|_{\pi}$ is bounded above by $0$ and below by $-\frac{25}{2}.$
\end{example} 

Combining Theorem \ref{thm:sens} with the proof of Corollary \ref{Cor:minpst}, we obtain the connected unweighted $n$-vertex graphs and the unit vectors admitting PST that attain the least value of $\frac{d^2f}{dt^2}\bigg|_{\tau}$. These unit vectors have the most sensitive fidelity with respect to the PST time, and these are precisely those that minimize the PST time amongst all connected unweighted $n$-vertex graphs.

\begin{corollary}
\label{cor:sens1} 
Amongst all connected unweighted $n$-vertex graphs that admit perfect state transfer between unit vectors at time $\tau$ relative to $L$ (respectively, relative to $A$), the least value of $\frac{d^2f}{dt^2}\bigg|_{\tau}$ is attained by the join graph in Theorem \ref{Thm:minperjoin}(1) (respectively, Theorem \ref{Thm:minperjoin}(2)) and unit vectors $\bx$ and $\by$, where $\by$ is given in Corollary \ref{Cor:minpst}(1) (respectively, Corollary \ref{Cor:minpst}(2)).
\end{corollary}

\bigskip

\noindent \textbf{Acknowledgement.}
C.\ Godsil is supported by NSERC grant no.\ RGPIN-9439. S.\ Kirkland is supported by NSERC grant no.\ RGPIN-2025-05547. H.\ Monterde is supported by the University of Manitoba Faculty of Science and Faculty of Graduate Studies. We thank the Department of Mathematics and the Graduate Mathematics Society at the University of Manitoba for supporting the research visit of C.\ Godsil, where we all started to work on this project.

\bibliographystyle{alpha}
\bibliography{mybibfile}
\end{document}